\documentclass[12pt,draftcls,onecolumn]{IEEEtran}
\usepackage{amsthm,amssymb,amsmath}
\usepackage{graphicx}
\usepackage[]{authblk}

\newcommand{\ra}{\rightarrow}
\newcommand{\beqano}{\begin{eqnarray*}}
\newcommand{\eeqano}{\end{eqnarray*}}
\newcommand{\beqa}{\begin{eqnarray}}
\newcommand{\eeqa}{\end{eqnarray}}

\newcommand{\bE}{\mathbf{E}}

\newtheorem{definition}{Definition}
\newtheorem{theorem}{Theorem}
\newtheorem{lemma}{Lemma}
\newtheorem{corollary}{Corollary}

\allowdisplaybreaks

\begin{document}

\title{~Erasure Multiple Descriptions~}
\date{}

\author{Ebad Ahmed}
\author{Aaron B. Wagner}
\affil{School of Electrical and Computer Engineering\\
Cornell University, Ithaca, NY 14853, USA}


\pagestyle{headings} \maketitle

\begin{abstract}
We consider a binary erasure version of the $n$-channel multiple descriptions problem with symmetric descriptions, i.e., the rates of the $n$ descriptions are the same and the distortion constraint depends only on the number of messages received. We consider the case where there is \emph{no excess rate} for every $k$ out of $n$ descriptions, i.e., any subset of $k$ messages has a total rate of $R(D_k) = 1 - D_k$, where $R(\cdot)$ is the Shannon rate-distortion function and $D_k$ is the distortion constraint when $k$ descriptions are received at the decoder. Our goal is to characterize the achievable distortions $D_1, D_2,\ldots,D_n$. We measure the fidelity of reconstruction using two distortion criteria: an \emph{average-case} distortion criterion, under which distortion is measured by taking the average of the per-letter distortion over all source sequences, and a \emph{worst-case} distortion criterion, under which distortion is measured by taking the \textit{maximum} of the per-letter distortion over all source sequences. We present achievability schemes, based on random binning for average-case distortion and systematic MDS (\emph{maximum distance separable}) codes for worst-case distortion, and prove optimality results for the corresponding achievable distortion regions. We then use the binary erasure multiple descriptions setup to propose a layered coding framework for multiple descriptions, which we then apply to vector Gaussian multiple descriptions and prove its optimality for symmetric scalar Gaussian multiple descriptions with two levels of receivers and no excess rate for the central receiver. We also prove a new outer bound for the general multi-terminal source coding problem and use it to prove an optimality result for the robust binary erasure CEO problem. For the latter, we provide a tight lower bound on the distortion for $\ell$ messages for any coding scheme that achieves the minimum achievable distortion for $k \le \ell$ messages.
\end{abstract}

\section{Introduction}
While the information-theoretic study of network capacity has played a pivotal role in the development of wireless communications~\cite{tse_visw:wireless}, network rate-distortion theory has had a much smaller impact on the design of practical systems. The reason for this is arguably two-fold. First, the mathematically challenging nature of network source coding has hindered progress toward understanding the fundamental limits of lossy data compression. The rate regions of many important network source coding problems have yet to be characterized and solutions for even simple networks are analytically involved. Second, prominent network source coding problems often are poor models that abstract away key properties of practical systems. In particular, such models often fail to accurately capture the distortion resulting from source quantization in practical systems.

This paper attempts to circumvent these two issues by focusing on the use of the erasure distortion measure~\cite[p.~370]{Cover:IT} for a binary source. The erasure distortion measure is well-suited for digital sources since it does not permit the decoder to make errors in its reconstruction of the source, but allows it to declare an erasure for any source symbol about which it is uncertain. Errors in digital data streams generally wreak havoc unless detailed knowledge of the digital representation is used to minimize their impact. Erasures, however, are tolerable since they can be detected by higher-level applications, which can either interpolate to fill in the missing data or wait until enough data is received to correct all of the erasures. Erasure formulations should also be useful as starting points for the design of practical codes for network rate-distortion. In the theoretical development of modern channel codes like LDPC, many of the code designs and performance characterizations were first established for the erasure channel~\cite{Rich_Urb:mod_cod_th}.

This paper looks at the binary erasure version of an important network source coding problem, the multiple descriptions (MD) problem~\cite{elg_cover:ach_md}-\cite{chen:gmdrric}. Multiple descriptions is a source coding technique in which multiple encoded descriptions of a single source sequence are sent to the decoder over separate channels. This is an effective way to deal with channel failure and packet loss in packet networks, particularly in the case where retransmission of lost packets is not feasible (e.g., audio/video streaming) and the decoder must reconstruct the source with only the packets it has successfully received. The MD problem also constitutes a reasonable model for transmission of digital data (images, video, and sound) over peer-to-peer networks.

An important regime within MD is that of \emph{no excess rate}, i.e., the sum rate required to achieve distortion $D$ at the receiver equals $R(D)$, where $R(\cdot)$ is the Shannon rate-distortion function. This is a useful regime to study, since it allows us to not sacrifice end-to-end performance for intermediate performance (i.e., when the number of received descriptions is less than the number required to achieve distortion $D$). For most sources, the no excess rate regime is characterized by poor intermediate performance (e.g.,~\cite{ozarow:md}): if a coding scheme is near-optimal for $k$ receptions, it often yields high distortions for $m < k$ receptions. For binary erasure MD, however, it is possible to obtain good intermediate performance under no excess rate.

\subsection{Results}
We focus on binary erasure MD with no excess rate for every $k$ out of $n$ descriptions, i.e., any subset consisting of $k$ messages must have a total rate of $R(D_k)$, where $D_k$ is the distortion constraint the decoder must obey when $k$ messages are received. We consider symmetric descriptions, i.e., the rates of the $n$ descriptions are the same and the distortion constraint depends only on the number of messages received. In fact, no excess rate implies symmetric descriptions for $k<n$: if every $k$ out of $n$ descriptions have sum rate $R(D_k)$, then each rate must be $R(D_k)/k$. We examine two distortion criteria; an \emph{average-case} distortion criterion, which measures the reconstruction fidelity by the average of the per-letter distortion over all source sequences, and a \emph{worst-case} distortion criterion, which measures the reconstruction fidelity by the maximum of the per-letter distortion over all source sequences. The average-case criterion is the standard criterion used in the literature. The worst-case criterion is less commonly used but arguably more appropriate in this setting. It is a universal distortion measure and is insensitive to the source model since it does not a require a source distribution. Our main contributions are:
\begin{enumerate}
\item applying the binary erasure model to multiple description coding and focusing on the worst-case distortion criterion,
\item proposing, for all $n$ and $k$, coding schemes for both average-case and worst-case distortion criteria and characterizing their achievable distortion region when $m \le k$ descriptions are received at the decoder. The scheme for average-case distortion is based on random binning and can be viewed as of a concatenation of $(n,1)$ and $(n,k)$ source-channel erasure codes~\cite{puri:nkcodes_2}. The scheme for worst-case distortion is a practical zero-error coding scheme based on MDS (\emph{maximum distance separable}) codes.
\item providing, for both average-case and worst-case distortion criteria, a tight lower bound on the distortion when a single message is received at the decoder. For worst-case distortion, the outer bound holds for all $n$ and $k$. Moreover, we show that the MDS coding scheme is Pareto optimal in the achievable distortions $D_1,\ldots,D_k$ for all $n$ and $k$, and, for certain ranges of $n$ and $k$, is also optimal when more than one message is received at the decoder. For average-case distortion, our outer bound holds, modulo a closure operation, for all $n$ and $k$ satisfying $\left(1 - \frac{1}{n}\right)^k \le \frac{1}{2}$. In addition, for $n > 3$ and $k=2$, we provide an outer bound on the optimal single-message distortion that differs by exactly $1/n$ from the distortion achieved by the random binning scheme. Our results for the special case in which there is no distortion for $k$ messages (i.e., any $k$ messages allow the decoder to construct the original source sequence completely) have appeared in \cite{ahmed_wagner:itw09} (average-case distortion) and \cite{ahmed_wagner:isit09} (worst-case distortion).
\item proposing a coding scheme, based on the binary erasure MD coding schemes, for vector Gaussian MD and showing that it is optimal for scalar Gaussian MD with two levels of receivers and no excess rate for the central receiver. The scheme involves quantizing the vector Gaussian source according to a given quadratic distortion constraint and then transmitting the quantized version over the $n$ channels according to the aforementioned binary erasure coding schemes. This shows that the binary erasure coding schemes can be used as part of a more general, layered coding scheme for multiple descriptions with a generic source distribution and arbitrary distortion metric.
\item proving a new outer bound for the general multi-terminal source coding problem that improves upon the outer bound in \cite{wagnervenkat05}, and
\item providing, for the robust binary erasure CEO problem with symmetric rates, a tight lower bound on the distortion for $\ell$ messages for any coding scheme that achieves the minimum achievable distortion for $k \le \ell$ messages. The robust binary erasure CEO problem is a generalization of MD in that the encoders observe erased versions of the source instead of the source itself. This problem constitutes a reasonable model for decentralized peer-to-peer networks in which peers can generate new descriptions based on their partial copies of the source file.
\end{enumerate}

\subsection{Relation to Prior Work}
An achievable rate region for the 2-description MD problem was first provided by El Gamal and Cover~\cite{elg_cover:ach_md}. This region was shown to be tight for a scalar Gaussian source and quadratic distortion measure by Ozarow~\cite{ozarow:md}, and for a discrete memoryless source (DMS) with no excess rate for two descriptions by Ahlswede~\cite{ahlswede:md_ner}. Zhang and Berger~\cite{zhan_ber:md} obtained a rate region for the 2-description case that contained points strictly outside the El Gamal-Cover rate region. Venkataramani, Kramer and Goyal provided a rate region for the $n$-description case~\cite{ven_kra_goy:md}, which was improved upon by Pradhan, Puri, and Ramchandran~\cite{pradpuriramc04,puri:nkcodes_2}. Tian and Chen proposed a coding scheme for the $n$-description case, with symmetric rates and distortion constraints, that combined a channel coding component with a source coding component to attain rate-distortion points outside the region proposed in \cite{pradpuriramc04} in the Gaussian case~\cite{tian_chen:md_mds}. Wang and Viswanath derived the minimal achievable sum rate for vector Gaussian MD with individual and central receivers~\cite{wang_visw:vect_gmdic}. More recently, Chen characterized the rate region of scalar Gaussian MD with individual and central distortion constraints~\cite{chen:gmdrric}.

Multiple descriptions with no excess rate is a generalization of the problem of successive refinement~\cite{koshelev:succ_ref,equitz_cover:succ_ref,rimoldi:succ_ref}, in which descriptions received in addition to the minimum number required to reconstruct the source with a given distortion are used to improve the quality of reconstruction. The MD problem is also similar to the problem of lossy packet transmission considered by Albanese \emph{et al.}~\cite{albanese:pet}. They propose a coding method to deal with packet loss in erasure networks that involves assigning a priority level to messages. The messages are encoded into packets, and the priority level determines the minimum number of packets required to reconstruct the message. Other work on similar problems include symmetric multi-level diversity (MLD) coding~\cite{roche:sym_mld}, in which $K$ sources, each with a different level of importance, are encoded by $K$ encoders. The decoders have access to only a subset of the encoded descriptions, and each decoder attempts to reconstruct the $k$ most important sources, where $k$ is the number of descriptions that are accessible to it. More recently, Mohajer \emph{et al.}~\cite{mohajer:asym_mld} have considered a variation on symmetric MLD coding in which $2^K-1$ sources are encoded by $K$ encoders, and have characterized the rate region for $K=3$.

Our binary erasure MD problem with no excess rate and no distortion for every $k$ out of $n$ messages is particularly significant in the context of peer-to-peer networks, since it can be used to study the tradeoff between the performance of two competing technologies: fountain codes~\cite{luby:lt_codes,shokrollahi:raptor_codes} and BitTorrent~\cite{cohen:BitTorrent}. For large $n$ and small $k$, the MD problem mimics rateless fountain codes, since out of a large number of descriptions, only a handful must be received (collected) in order to construct the source with zero distortion. Fountain codes are known to work well in erasure networks, but they usually have poor intermediate performance. Sanghavi~\cite{sanghavi:rateless_codes} provides an outer bound for rateless codes on the fraction of source symbols that can be decoded as a function of the number of encoded symbols received. For $k = n$, the MD problem resembles the BitTorrent, where all of the relevant packets must be received to allow for complete reconstruction of the source. The BitTorrent provides good intermediate performance but suffers from the ``coupon collector" problem; the initial pieces of the source can be acquired relatively rapidly, but it takes much longer to collect the final pieces. By varying $n$ and $k$ in the binary erasure MD model, therefore, the middle ground between fountain codes and the BitTorrent can be explored.

The rest of this paper is organized as follows. In Section~\ref{sec:prob_form}, we formulate the $n$-channel binary erasure MD problem. Sections~\ref{sec:average-case} and \ref{sec:worst-case} are devoted to our results for average-case distortion and worst-case distortion, respectively. In Sections~\ref{sec:gauss_md} and \ref{sec:dec_encoding}, we describe our results for vector Gaussian MD and the robust binary erasure CEO problem, respectively.

\section{The $n$-channel Binary Erasure Multiple Descriptions Problem}
\label{sec:prob_form}
Let $\{X_t\}_{t=1}^\infty$ be a memoryless uniform binary source, with the random variables $X_t$ taking values in the alphabet $\mathcal{X}=\{+,-\}$. Let $\mathcal{\hat{X}}$ be the reconstruction space $\{+,-,0\},$ where $0$ denotes the erasure symbol, with an associated distortion measure $d:\mathcal{X} \times \mathcal{\hat{X}} \rightarrow \{0,1,\infty\}$ such that
\begin{align*}
d(x,\hat{x}) =
\begin{cases}
0 & \text{if $\hat{x} = x$} \\
1 & \text{if $\hat{x} = 0$} \\
\infty & \text{otherwise}.
\end{cases}
\end{align*}
The above per-letter measure is known as the erasure distortion measure ~\cite[p.~370]{Cover:IT}. A \emph{encoder} is a function $f_i^{(l)}:\mathcal{X}^l \rightarrow \{1,\ldots,M_i^{(l)}\}$. A \emph{decoder} is a function $g_\mathcal{K}^{(l)}: \prod_{k \in \mathcal{K}} \{1,\ldots,M_k^{(l)}\}  \rightarrow \hat{\mathcal{X}}^l$, where $\mathcal{K}$ is the set of descriptions received.

Let $\mathcal{N}=\{1,\ldots,n\}.$ The $n$-channel multiple descriptions problem, illustrated in Figure~\ref{fig:model}, can be formulated as follows. There are $n$ encoders. Encoder $f_i^{(l)}, \ i \in \mathcal{N}$, encodes and transmits a description of a length-$l$ source sequence $x^l$ over channel $i$. The receiver either receives this description without errors or it does not receive it at all. Excluding the case where none of the descriptions is received, the receiver may receive $2^n - 1$ different combinations of the $n$ descriptions. Thus it can be represented by the $2^n-1$ decoding functions $g_\mathcal{K}^{(l)}, \ \mathcal{K} \subseteq \mathcal{N}$, $\mathcal{K} \neq \emptyset$. Based on the set of descriptions received, the receiver employs the corresponding decoding function to output a reconstruction of the original source string subject to a distortion constraint. We consider symmetric descriptions, i.e., each description has the same rate and the distortion constraint depends only on the number of descriptions received.

We measure the fidelity of the reconstruction using two distortion criteria: an \emph{average-case} distortion criterion, under which distortion is measured by taking the average of the per-letter distortion over all source sequences, and a \emph{worst-case} distortion criterion, under which distortion is measured by taking the \textit{maximum} of the per-letter distortion over all source sequences. We define \emph{achievability} for the two criteria as follows. Let $\hat{X}_\mathcal{K}^l = g_\mathcal{K}^{(l)}(\{f_k^{(l)}(X^l): k \in \mathcal{K}\})$ be the reconstruction sequence corresponding to the source sequence $X^l$.

\begin{definition}[\textbf{Average-case distortion}]
\label{def:achievability_avg-case}
The rate-distortion vector $(R,D_1,\ldots,D_n)$ is \emph{achievable} if for some $l$ there exist encoders $f_i^{(l)}$, $i\in \mathcal{N}$ and decoders $g_\mathcal{K}^{(l)}$, $\mathcal{K}\subseteq \mathcal{N}$, $\mathcal{K} \neq \emptyset$, such that\footnote{All logarithms and exponentiations in this paper have base $2$ unless explicitly stated.}
\begin{align*}
R &\ge \frac{1}{l}\log M_i^{(l)} \ \textrm{ for all $i$, and} \\
D_k &\ge \max_{\mathcal{K}:|\mathcal{K}| = k} \textbf{E}\left[\frac{1}{l}\sum_{t=1}^l d(X_t,\hat{X}_{\mathcal{K},t})\right].
\end{align*}
\end{definition}
We use $\mathcal{RD}_{avg}$ to denote the set of achievable rate-distortion vectors and $\overline{\mathcal{RD}}_{avg}$ to denote its closure.

\begin{definition}[\textbf{Worst-case distortion}]
\label{def:achievability_worst-case}
The rate-distortion vector $(R,D_1,\ldots,D_n)$ is \emph{achievable} if for some $l$ there exist encoders $f_i^{(l)}$, $i\in \mathcal{N}$ and decoders $g_\mathcal{K}^{(l)}$, $\mathcal{K}\subseteq \mathcal{N}$, $\mathcal{K} \neq \emptyset$, such that
\begin{align*}
R &\ge \frac{1}{l}\log M_i^{(l)} \ \textrm{ for all $i$, and} \\
D_k &\ge \max_{\mathcal{K}:|\mathcal{K}| = k} \max_{x^l \in \mathcal{X}^l}\left[\frac{1}{l}\sum_{t=1}^l d(X_t,\hat{X}_{\mathcal{K},t})\right].
\end{align*}
\end{definition}
We use $\mathcal{RD}_{worst}$ to denote the set of achievable rate-distortion vectors. We describe our results for average-case distortion in the next section and for worst-case distortion in Section~\ref{sec:worst-case}. For both distortion criteria, we consider the case where there is \emph{no excess rate} for every $k$ out of $n$ descriptions, i.e., $kR = R(D_k) = 1 - D_k$, where $R(\cdot)$ is the Shannon rate-distortion function. Thus $R = (1-D_k)/k$. We will henceforth use $R$ to denote $(1-D_k)/k$. Our goal is to characterize the achievable distortions $D_1,\ldots,D_n$ for both distortion criteria.

\begin{figure}[htp]
\centering
  \includegraphics[width=3.5in]{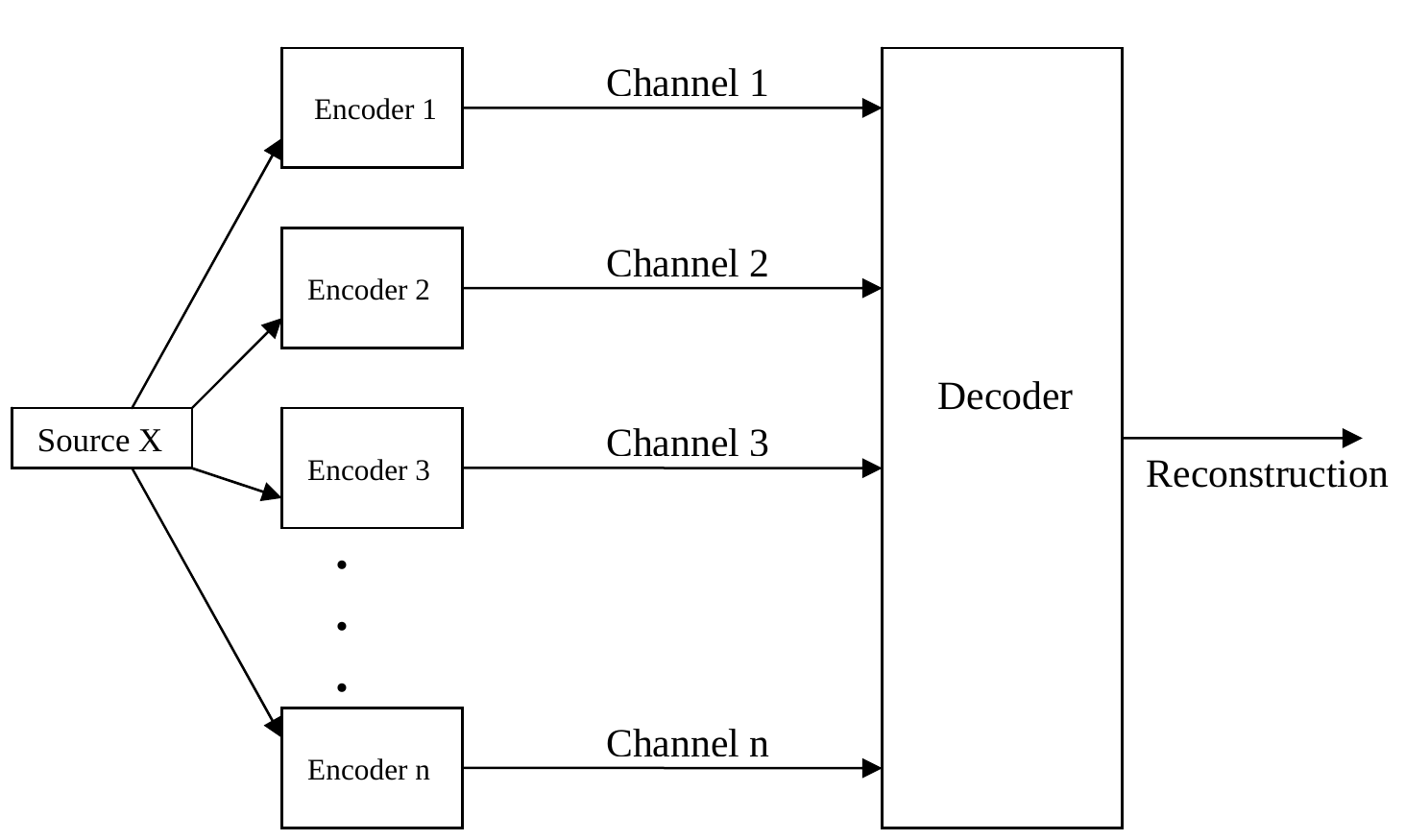}
\caption{The $n$-channel multiple descriptions problem}\label{fig:model}
\end{figure}

It should be pointed out that the $k=n$ case is particularly simple. Let $D_i, \ i \in \mathcal{N}$ be the distortion constraint when the receiver receives $i$ messages. No excess rate for $n$ descriptions dictates that the sum-rate of the $n$ messages is exactly $(1-D_n)$, which in turn implies that the rate of each message is $(1-D_n)/n$. The problem then reduces to characterizing the optimal $D_1,\ldots,D_n$. Consider a coding scheme that takes a source string of length $l$ and erases the last $lD_n$ bits. The remaining $l(1-D_n)$ bits are divided into $n$ disjoint parts, each consisting of $l(1-D_n)/n$ bits. Encoder $i$ transmits the $l(1-D_n)/n$ bits in the $i^{th}$ part to the decoder over the $i^{th}$ channel, with erasures in places of the remaining $l - l(1-D_n)/n$ bits. Thus upon reception of any $k$ descriptions, the decoder can reconstruct $kl(1-D_n)/n$ bits of the original source string. Clearly, this scheme achieves $D_k = 1-k(1-D_n)/n$ under both the average-case and worst-case distortion criteria. Moreover, for any code that achieves the rate-distortion vector $(1-D_n/n,D_1,\ldots,D_n)$, every description has rate $(1-D_n)/n$ and therefore any set of $k$ message can reveal no more than a fraction $k(1-D_n)/n$ bits of the original source string. Thus $$\max_{\mathcal{K}:\mathcal{K}=k} \textbf{E}\left[\frac{1}{l}\sum_{t=1}^l d(X_t,\hat{X}_{\mathcal{K},t})\right] \ge 1 - k(1-D_n)/n,$$ and $$\max_{\mathcal{K}:\mathcal{K}=k} \max_{x^l \in \mathcal{X}^l}\left[\frac{1}{l}\sum_{t=1}^l d(X_t,\hat{X}_{\mathcal{K},t})\right] \ge 1 - k(1-D_n)/n.$$ Thus the aforementioned coding scheme achieves the optimal $D_1,\ldots,D_n$ under both the average-case and worst-case distortion criteria.

We use the insight obtained from the $k = n$ case to construct codes for the more complicated case in which $k < n$. No excess rate for a particular set of $k$ descriptions requires that information transmitted over the corresponding channels be independent. Since we impose no excess rate for every size-$k$ subset of descriptions, information transmitted over any $k$ channels must be mutually independent. The coding scheme for $k = n$ ensures that this condition is met by dividing an erased version of the source string into $n$ disjoint (and therefore independent) parts and transmitting them uncoded over the $n$ channels. This strategy of sending independent uncoded bits works as long as the bits transmitted over each channel are disjoint. In particular, if $R = (1-D_k)/k \le 1/n$ (equivalently, $D_k \ge 1-k/n$), the source string can always be divided into $n$ disjoint, equal parts, each containing a fraction $R$ of the total number of bits. If $D_k < 1-k/n$, however, then $R > 1/n$ and it is not possible to divide the source string into $n$ disjoint parts each containing a fraction $R$ of the total number of bits, since each part must then contain more than $1/n$ of the total number of bits. Transmitting uncoded bits, therefore, will only be optimal for a rate up to $1/n$ only; in order to achieve a rate larger than $1/n$, additional information about the source must be transmitted along with each description, and this information must be mutually independent for every set of $k$ descriptions.

The threshold $D_k = 1 - k/n$ therefore plays an important role in our coding schemes for both average-case and worst-case distortions. If $D_k \ge 1 - k/n$, our coding scheme is based solely on the transmission of independent uncoded bits over the $n$ channels as described above. If $D_k < 1-k/n$, then in addition to sending uncoded bits, we employ random binning (for average-case distortion) and MDS codes (for worst-case distortion) to communicate additional information about the source sequence. The random binning component works by randomly binning an erased version of all possible source sequences at each encoder. Each encoder transmits uncoded bits from the observed source sequence along with the bin index of the corresponding erased version. The decoder uses the uncoded bits and the bin indices to output a partial reconstruction of the source sequence. Decoding the binned erased version in particular allows the decoder to reconstruct source bits other than the ones it receives uncoded. The average-case distortion scenario is conceptually simple, but provides weaker guarantees on optimality. The MDS coding scheme for worst-case distortion is based on a similar idea (transmission of uncoded bits plus encoded information about an erased version of the source string), but as we will see later, the worst-case distortion scenario provides much stronger guarantees on optimality than average-case distortion. The coding schemes for average-case and worst-case are described in detail in Sections~\ref{sec:achievability_avg-case} and \ref{sec:worst-case-achievability}, respectively.

\section{The Average-case Distortion Criterion}
\label{sec:average-case}

\subsection{An Achievability Result}
\label{sec:achievability_avg-case}

\begin{definition}
Given $n$, $k \le n$, and $D_k \in [0,1]$, define
\begin{align*}
\tilde{\textbf{R}} &= \left(R,1-R,1-2R,\ldots,1-(k-1)R,D_k,D_k-R,D_k-2R,\ldots,D_k-(n-k)R\right), \ \text{ and} \\
\hat{\textbf{R}} &= \left(R,1-\frac{1}{n},1-\frac{2}{n},\ldots,1-\frac{k-1}{n},D_k,\left(\frac{n-k-1}{n-k}\right)D_k,\left(\frac{n-k-2}{n-k}\right)D_k,\ldots,\left(\frac{1}{n-k}\right)D_k,0\right).
\end{align*}
\end{definition}
The following theorem shows that it is possible to achieve good intermediate performance when $m < k$ descriptions are received at the decoder.
\begin{theorem}
\label{thm:achievability_avg-case}
Let $D_k \in [0,1]$. For any $n$ and $k \le n$, if $D_k \ge 1-\frac{k}{n}$, then $\tilde{\textbf{R}} \in \overline{\mathcal{RD}}_{avg}$. If $D_k < 1-\frac{k}{n}$, then $\hat{\textbf{R}} \in \overline{\mathcal{RD}}_{avg}$.
\end{theorem}
\begin{proof}
\textbf{Case I}: $D_k \ge 1-\frac{k}{n}$
\newline
Assume without loss of generality that $D_k$ is rational (if $D_k$ is irrational, then we can prove achievability for a sequence of rational distortions in $[1-k/n,1]$ converging to $D_k$ and take limits). Then there exists a positive integer $l'$ such that $l'R$ is a positive integer. Choose a blocklength $l=\alpha nl'$, where $\alpha$ is any positive integer. Observe a length-$l$ source sequence $X^l$, and divide $X^l$ into $n$ disjoint parts such that each part contains $l/n=\alpha l'$ bits. (The division is the same regardless of the source realization.) Label the parts $X_i, \ i \in \mathcal{N}$. Choose $lR$ bits from each of the $n$ parts (since $D_k \ge 1-\frac{k}{n}$, $lR \le \frac{l}{n}$ and therefore $lR$ bits can be chosen from each part). Denote by $y_i$ the set of $lR$ bits chosen from $X_i$. Transmit $y_i$ uncoded over the $i^{th}$ channel.

The decoding is trivial. If $m$ descriptions, say $(y_1,\ldots,y_m)$, are received, output $\hat{X}_m^l$ as the reconstruction of $X^l$, where $\hat{X}_m^l$ is such that the $mlR$ bits corresponding to $(y_1,\ldots,y_m)$ are non-erased and the other ($l- mlR)$ bits are erasures. The distortion, therefore, is $(l- mlR)/l = 1 - mR$. When $k$ descriptions are received, the distortion is $1-kR = D_k$. Thus the rate-distortion vector $(R,1-R,1-2R,\ldots,1-(k-1)R,D_k,D_k-R,D_k-2R,\ldots,D_k-(n-k)R) \in \mathcal{RD}_{avg}$, and therefore also lies in $\overline{\mathcal{RD}}_{avg}$.
\newline
\newline
\textbf{Case II}: $D_k < 1-\frac{k}{n}$
\newline
The scheme for this case is an extension of the scheme for Case I. It has two components; random binning and transmission of uncoded source bits. An erased version of every source sequence is binned separately at each encoder. The observed source string is divided into $n$ disjoint parts. Each uncoded part is then sent on one of the $n$ channels along with the corresponding bin index of the erased version of the source. If less than $k$ descriptions are received, the decoder outputs a partial reconstruction based solely on the uncoded parts; if $k$ or more descriptions are received, the decoder outputs a reconstruction based on the uncoded parts and the bin indices.

Assume again that $D_k$ is rational. Choose $\epsilon > 0$, and define $R' = (1-D_k)/k - 1/n + \epsilon$. Since $D_k$ is rational, there exists a positive integer $l'$ such that $l'D_k/(n-k)$ is an integer. Choose a blocklength $l=\alpha nl'$, where $\alpha$ is any positive integer.

\emph{Random binning:} Construct $n$ sets of bins such that every set contains $2^{lR'}$ bins. For every length-$l$ source string $x^l \in \mathcal{X}^l$, construct an erased version as follows. Divide $x^l$ into $n$ disjoint parts such that each part contains $l/n=\alpha l'$ bits (the division is done identically for all source sequences). For each part, replace the last $lD_k/(n-k)$ bits by erasures (since $D_k < 1-\frac{k}{n}$, each part contains $l/n > lD_k/(n-k)$ bits). Assign the resulting erased version $x_e^l$ uniformly at random, and independently from other strings, to one of the $2^{lR'}$ bins in the $i^{th}$ set, for all $i\in \mathcal{N}$. The assignment is done only once for each erased version. This is important because multiple source strings can have the same erased version. Denote the assignments by $\Gamma_i$.

\emph{Encoding:} Let $X^l$ be the observed source sequence. Divide $X^l$ into $n$ disjoint parts each containing $l/n$ bits as described above. Label the parts $X_i$, $i \in \mathcal{N}$. Let $b_i = \Gamma_i(X^l)$ be the index of the bin containing the erased version of $X^l$ in the $i^{th}$ bin set. Transmit $(X_i,b_i)$ over the $i^{th}$ channel.

\emph{Decoding:} If $m$ descriptions, say $\{(X_1,b_1),\ldots,(X_m,b_m)\}$, are received, where $m < k$, output $\hat{X}_m^l$ as the reconstruction of $X^l$, where $\hat{X}_m^l$ is such that the $ml/n$ bits corresponding to $(X_1,\ldots,X_m)$ are non-erased and the other ($l- ml/n)$ bits are erasures. If $m > k$ descriptions are received, say $\{(X_1,b_1),\ldots,(X_m,b_m)\}$, choose any $k$ descriptions, say $\{(X_1,b_1),\ldots,(X_k,b_k)\}$, and search the bins $(b_1,\ldots,b_k)$ for a sequence $y$ such that $\Gamma_i(y)=b_i$, $i = 1,\ldots,k$, and $y$ is consistent with the partially revealed source string $(X_1,\ldots,X_k)$. Output $\hat{X}_m^l = \{(X_1,\ldots,X_m)\} \cup \{y\}$ as the reconstruction of $X^l$. (Thus the non-erased bits in $\hat{X}_m^l$ are the bits revealed by $(X_1,\ldots,X_m)$ or by the erased version $y$, or both.) There is guaranteed to be at least one such sequence $y$ in the bins indexed by $b_1,\ldots,b_k$. If there is more than one such sequence, output the all-erasure string as the reconstruction of $X^l$. (This will suffice to meet our distortion constraint.)

\emph{Error analysis:} We say an error $E_\mathcal{S}$ has occurred at the decoder if, for a set $\mathcal{S} = \{s_1,\ldots,s_k\}$ of $k$ descriptions, there exists an erased version $y\neq X_e^l$ such that $\Gamma_{s_i}(y)=\Gamma_{s_i}(X_e^l)$ for all $s_i \in \mathcal{S}$ and $y$ is consistent with $(X_{s_1},\ldots,X_{s_k})$. Let $\mathcal{C_S}$ be the set of erased versions that are consistent with $(X_{s_1},\ldots,X_{s_k})$. Define $E = \bigcup_{\mathcal{S}, |\mathcal{S}| = k} E_\mathcal{S}$.
We bound $\Pr(E)$ as follows.
\begin{align*}
\Pr(E) &\le \sum _{\mathcal{S}, |\mathcal{S}| = k} \Pr(E_\mathcal{S}) \\
&= \sum_{\mathcal{S}, |\mathcal{S}| = k} \Pr(\exists y\neq X_e^l,y \in \mathcal{C_S}: \Gamma_{s_i}(y)= \Gamma_{s_i}(X_e^l) \ \forall s_i \in \mathcal{S}) \\
&= \sum_x p(x) \sum_{\mathcal{S}, |\mathcal{S}| = k} \Pr(\exists y\neq x_e^l, y \in \mathcal{C_S}: \Gamma_{s_i}(y)= \Gamma_{s_i}(x_e^l))\\
&\le \sum_x p(x) \sum_{\mathcal{S}, |\mathcal{S}| = k} \sum_{\substack{y\neq x_e^l \\ y\in \mathcal{C_S}}} \Pr(\Gamma_{s_i}(y)= \Gamma_{s_i}(x_e^l) \ \forall s_i \in \mathcal{S}) \\
&\le \sum_x p(x) \sum_{\mathcal{S}, |\mathcal{S}| = k} 2^{-klR'} |\mathcal{C_S}| \\
&= \sum_x p(x) \sum_{\mathcal{S}, |\mathcal{S}| = k} 2^{-kl(\frac{1-D_k}{k}-\frac{1}{n}+\epsilon)}\cdot 2^{(n-k)(\frac{l}{n}-l\frac{D_k}{n-k})} \\
&= \sum_x p(x) \sum_{\mathcal{S}, |\mathcal{S}| = k} 2^{-lk\epsilon}\\
&\le { n \choose k } 2^{-lk\epsilon}.
\end{align*}
We now show that for any $\epsilon > 0$, the $(n+1)$-tuple $(R+\epsilon,1-\frac{1}{n}+\epsilon,1-\frac{2}{n}+\epsilon,\ldots,1-\frac{k-1}{n}+\epsilon,D_k+\epsilon,(\frac{n-k-1}{n-k})D_k+\epsilon,(\frac{n-k-2}{n-k})D_k+\epsilon,\ldots,(\frac{1}{n-k})D_k+\epsilon,\epsilon)$ is achievable, and thus $(R,1-\frac{1}{n},1-\frac{2}{n},\ldots,1-\frac{k-1}{n},D_k,(\frac{n-k-1}{n-k})D_k,(\frac{n-k-2}{n-k})D_k,\ldots,(\frac{1}{n-k})D_k,0) \in \overline{\mathcal{RD}}_{avg}$. Fix $\epsilon > 0$ and define $R'$ as above. In our scheme, any description $(X_i,b_i)$ has rate $R = 1/n + R'$, where $1/n$ is the rate due to $X_i$ and $R'$ is the rate due to binning. Thus $R = 1/n + (\frac{1-D_k}{k} - 1/n + \epsilon) = (1-D_k)/k + \epsilon$. Moreover, if $m < k$ descriptions are received, the decoder outputs $ml/n$ bits as revealed by the $m$ descriptions and the other $(l-ml/n)$ bits as erasures. Thus $D_m = 1 - m/n < 1 - m/n + \epsilon$. If $k$ descriptions are received, say $\mathcal{S}=\{s_1,\ldots,s_k\}$, the decoder either outputs an erased version of the correct source sequence if $E_\mathcal{S}^c$ occurs, or outputs an all erasure string if $E_\mathcal{S}$ occurs. If $E_\mathcal{S}^c$ occurs, then the decoder receives $kl/n$ bits uncoded from the $k$ descriptions, and is able to figure out a further $(n-k)(l/n-lD_k/(n-k))=l(1-k/n-D_k)$ bits by using the bin indices to decode the erased version of the source sequence. Hence the maximum per-letter distortion over sets of $k$ descriptions is $1-(k/n+1-k/n-D_k)=D_k$ if $E^c$ occurs, and $1$ if $E$ occurs. Let $d_{\mathcal{S},x}$ be the per-letter distortion achieved using the set $\mathcal{S}$ of descriptions if the observed source string is $x^l$. Thus
\beqano
\bE_{f,g} \max_{\mathcal{S}, |\mathcal{S}| = k}  \bE_{X} [d_{\mathcal{S},X}] &\le& \bE_{f,g} \bE_{X} [\max_{\mathcal{S}, |\mathcal{S}| = k} d_{\mathcal{S},X}]  \\
&=& \bE_{f,g} \bE_{X} [1_E + D_k\cdot1_{E^c}] \\
&=& \Pr(E)+D_k(1-\Pr(E)) = (1-D_k)\Pr(E) + D_k\\
&\le& (1-D_k)\left[{ n \choose k } 2^{-kl\epsilon}\right]+ D_k,
\eeqano
which can be made smaller than $D_k + \epsilon$ by letting $\alpha\rightarrow \infty$. Thus $D_k + \epsilon$ is achievable for some sufficiently large $l$. If $m > k$ descriptions are received, then the decoder receives $ml/n$ bits uncoded, and is able to figure out a further $(n-m)(l/n-lD_k/(n-k))$ bits by decoding the binned erased version. Thus, if $E^c$ occurs, the maximum per-letter distortion is $1-m/n-((n-m)/n-(n-m)D_k/(n-k))=(\frac{n-m}{n-k})D_k$, and by the same analysis as above, a distortion of $(\frac{n-m}{n-k})D_k + \epsilon$ can be achieved for some sufficiently large $l$.
\end{proof}

\subsection{Optimality Results}
In this section we present optimality results for the random binning coding scheme described in the previous subsection. We first establish some preliminary results in Appendix~\ref{app:prelim} which will be used in the proofs of the following theorems. Our optimality results for the average-case deal deal primarily with single-message optimality, i.e., when only one message is received at the decoder. In the next section, we shall see that stronger optimality results can be established for the worst-case distortion criterion.

The following theorem shows that when only one message is received at the decoder, the scheme is optimal, modulo a closure operation, for all $n$ and $k$ satisfying $\left(1 - \frac{1}{n}\right)^k \le \frac{1}{2}$. Recall that, given $D_k$, we use $R$ to denote $(1-D_k)/k$.

\begin{definition}
For any fixed $D_k$, define $$D_1^* = \inf \{D_1:(R,D_1,\ldots,D_k,\ldots,D_n) \in \mathcal{RD}_{avg}\}.$$
\end{definition}

\begin{theorem}
\label{thm:conv_avg-case}
For any $n$ and $k \le n$, if $D_k \ge 1-\frac{k}{n}$, then for any $(R,D_1,\ldots,D_k,\ldots,D_n) \in \mathcal{RD}_{avg}$, $D_m \ge 1-mR$ for all $m \in \mathcal{N}$. If $D_k < 1-\frac{k}{n}$, $D_k$ is rational\footnote{For this theorem and subsequent theorems in this subsection, we consider rational values for $D_k$ since any code over a finite blocklength can yield only rational distortions.}, and $\left(1 - \frac{1}{n}\right)^k \le \frac{1}{2}$, then $D_1^* \ge 1 - \frac{1}{n}$.
\end{theorem}
\begin{proof}
See Appendix~\ref{app:conv_avg-case}.
\end{proof}
We note that $\left(1 - \frac{1}{n}\right)^k \le \frac{1}{2}$ implies $k \ge \frac{1}{\log \left(n/n-1\right)} := \lambda(n)$. Since $\lambda(n)/n \ra 1/\log e$ as $n \ra \infty$, the second part of Theorem~\ref{thm:conv_avg-case} provides a lower bound on $D_1^*$ for a large range of $k$ when $n$ is large.

The following theorem proves single-message optimality for the coding scheme when $n=4$ and $k=2$. This case is not included in Theorem~\ref{thm:conv_avg-case}.

\begin{theorem}
\label{thm:(4,2)_avg-case}
Let $D_k < 1-\frac{k}{n}$ and rational. If $n=4$ and $k=2$, then $D_1^* \ge 1 - \frac{1}{n}$.
\end{theorem}
\begin{proof}
See Appendix~\ref{app:(4,2)_avg-case}.
\end{proof}

Theorem~\ref{thm:conv_avg-case} handles the regime in which $k$ is large. We now study the other extreme, i.e., when $k$ is small. In particular, we look at the $k=2$ case. The following theorem provides a lower bound on the optimal single-message distortion for $n > 3$ and $k=2$. This lower bound differs from the distortion achieved by our coding scheme by exactly $1/n$, and thus becomes progressively tighter as $n$ increases.

\begin{theorem}
\label{thm:chebyshev_avg-case}
Let $D_k < 1-\frac{k}{n}$ and rational. If $k = 2$, then for $n > 3$, $D_1^* \ge 1 - \frac{2}{n}$.
\end{theorem}
\begin{proof}
See Appendix~\ref{app:chebyshev_avg-case}.
\end{proof}

We conjecture that the lower bound in Theorem~\ref{thm:chebyshev_avg-case} is not tight and that our scheme is in fact optimal. Evidence of this is provided by Theorem~\ref{thm:(4,2)_avg-case}.

\section{The Worst-case Distortion Criterion}
\label{sec:worst-case}

We turn now to the worst-case distortion criterion. We begin by presenting a practical, zero-error coding scheme based on systematic MDS codes that works for finite blocklengths. Like the random binning coding scheme for average-case distortion, the MDS coding scheme consists of two parts - uncoded bits and an MDS-code component. The uncoded component is similar to the uncoded component of the average-case coding scheme. The difference lies in the encoded component; instead of randomly binning an erased version of the source and then sending bin indices to the decoder (as the average-case distortion encoder does), the worst-case distortion encoder encodes the erased version using an $(n,k)$ systematic MDS code. The decoder outputs the uncoded bits and the bits revealed by the systematic part of the MDS code as the source reconstruction if less than $k$ descriptions are received. If $k$ or more descriptions are received, the decoder uses the uncoded bits and the bits revealed by the systematic part of the MDS code to decode the encoded erased version by applying an MDS decoding algorithm. The following subsection discusses the achievable distortion region of the MDS coding scheme.

\subsection{An Achievability Result}
\label{sec:worst-case-achievability}
\begin{theorem}
\label{thm:achievability_worst-case}
Let $D_k$ be a rational number in the interval $[0,1]$. For any $n$ and $k \le n$, if $D_k \ge 1-\frac{k}{n}$, then $\tilde{\textbf{R}} \in \mathcal{RD}_{worst}$. If $D_k < 1-\frac{k}{n}$, then $\hat{\textbf{R}} \in \mathcal{RD}_{worst}$.
\end{theorem}
\begin{proof}
\textbf{Case I}: $D_k \ge 1-\frac{k}{n}, \ D_k$ rational
\newline
Since $D_k$ is rational, there exists a positive integer $l'$ such that $l'R$ is a positive integer. Choose a blocklength $l=\alpha nl'$, where $\alpha$ is any positive integer. Observe a length-$l$ source sequence $X^l$, and divide $X^l$ into $n$ disjoint parts such that each part contains $l/n=\alpha l'$ bits. (The division is the same regardless of the source realization.) Label the parts $X_i, \ i \in \mathcal{N}$. Choose $lR$ bits from each of the $n$ parts (since $D_k \ge 1-\frac{k}{n}$, $lR \le \frac{l}{n}$ and therefore $lR$ bits can be chosen from each part). Denote by $y_i$ the set of $lR$ bits chosen from $X_i$. Transmit $y_i$ uncoded over the $i^{th}$ channel.

The decoding is trivial. If $m$ descriptions, say $(y_1,\ldots,y_m)$, are received, output $\hat{X}_m^l$ as the reconstruction of $X^l$, where $\hat{X}_m^l$ is such that the $mlR$ bits corresponding to $(y_1,\ldots,y_m)$ are non-erased and the other ($l- mlR)$ bits are erasures. Since the reconstruction sequence has $l- mlR$ erasures regardless of the source sequence, the worst-case distortion $D_m$ is $(l- mlR)/l = 1 - mR$. When $k$ descriptions are received, the worst-case distortion is $1-kR = D_k$. Thus the rate-distortion vector $(R,1-R,1-2R,\ldots,1-(k-1)R,D_k,D_k-R,D_k-2R,\ldots,D_k-(n-k)R) \in \mathcal{RD}_{worst}$.
\newline
\newline
\textbf{Case II}: $D_k < 1-\frac{k}{n}, \ D_k$ rational
\newline
For this case, we present an achievability scheme based on MDS (\emph{maximum distance separable}) codes\footnote{An $(n,k)$ MDS code is a linear code that satisfies the Singleton bound, i.e., the Hamming distance between any two codewords is $n-k+1$. Reed-Solomon codes, for instance, are MDS codes.}. Just as the achievability scheme for the average-case, this scheme has two components; uncoded bits and an MDS-code component. Let $m$ be the smallest integer such that $2^m \ge n$ and $\frac{mnk(n-k)}{n(1-D_k)-k}$ is an integer (such an $m$ exists because $D_k$ is rational). Define $q = 2^m$, and construct a $q$-ary MDS code of length $q-1$ and dimension $k$. By repeatedly puncturing this $(q-1,k)$ MDS code, we obtain a punctured MDS code of size $(n,k)$ \cite[p.~190]{Wicker:coding}. The punctured coordinates are revealed to the decoder. Let $\textbf{G}_1$ be the generator matrix of the punctured $(n,k)$ MDS code, and assume without loss of generality that $\textbf{G}_1$ is systematic, i.e., $\textbf{G}_1$ is of the form $[\textbf{I}_k|\textbf{A}]$, where $\textbf{I}_k$ is the $k \times k$ identity matrix and $\textbf{A}$ is a $k \times n-k$ matrix over the finite field GF($q$). Construct matrices $\textbf{G}_2,\ldots,\textbf{G}_n$ by shifting the columns of $\textbf{G}_1$ to the right, i.e., $\textbf{G}_i$ is the matrix formed by shifting the columns of $\textbf{G}_1$ by $i-1$ places, with the last $i-1$ columns of $\textbf{G}_1$ wrapping around. In particular, if $\textbf{G}_1 = [\textbf{I}_k|A_1\ldots A_n]$, where $A_1,\ldots,A_n$ are the columns of \textbf{A}, then $\textbf{G}_i = [A_{n-i+2}\ldots A_n|\textbf{I}_k|A_1\ldots A_{n-i+1}]$.

\emph{Encoding:} Let $X^l$ be the observed source string, of length $l = \frac{mnk(n-k)}{n(1-D_k)-k}$ bits. Divide $X^l$ into $n$ disjoint parts, each of length $\frac{mk(n-k)}{n(1-D_k)-k}$ bits. (The division is done the same way regardless of the source realization.) Let $X_i, \ i \in \mathcal{N}$ denote the last $lD_k/(n-k)$ bits of the $i^{th}$ part. Construct an erased version $X_e^l$ by replacing the last $lD_k/(n-k)$ bits in each of the $n$ parts by erasures. Thus $X_e^l$ has $l(1-\frac{nD_k}{n-k})= mnk$ bits. Each of the $n$ parts of $X_e^l$ has $mk$ bits and can therefore be treated as a concatenation of $k$ binary strings of length $m$, such that each of these binary strings is the binary representation of an element in GF($q$). Thus each of the $n$ parts of $X_e^l$ can be mapped to a vector of length $k$ in GF($q$). Label these vectors $p_j, \ j \in \mathcal{N}$. Let $y_j = p_j\textbf{G}_j, \ j \in \mathcal{N}$. Thus the $y_j$ are length-$n$ vectors in GF($q$). Let $y_{ji} = p_jG_{ji}$ denote the $i^{th}$ element of $y_j$ (here $G_{ji}$ is the $i^{th}$ column of $\textbf{G}_j$). Transmit ($X_i,y_{ji}: j \in \mathcal{N}$) over the $i^{th}$ channel.

\emph{Decoding:} Suppose $c < k$ descriptions are received at the decoder. Let $\mathcal{M} \subset \mathcal{N}$ denote the set of indices of the received descriptions. Assume without loss of generality that $i \in \mathcal{M}$. Thus the decoder receives $X_i$ and $y_{ji} = p_jG_{ji}: j \in \mathcal{N}$. Thus $lD_k/(n-k)$ bits are revealed to the decoder via $X_i$. Now for a fixed $i$, exactly $k$ of the $\textbf{G}_j, \ j \in \mathcal{N}$, (in particular, $\textbf{G}_{i-k+1},\ldots,\textbf{G}_i$) will have their $i^{th}$ column in the systematic part. Thus one symbol from $k$ of the $p_j, \ j \in \mathcal{N}$, can be decoded. By mapping these decoded symbols to their binary representations, the decoder can obtain a partial reconstruction of $X$. Let $\hat{X}_i$ represent the reconstructed source bits due to the $i^{th}$ description. Output $(\hat{X}_i:i \in \mathcal{M})$ as the reconstruction of $X^l$. If $m > k$ descriptions are received, then any $k$ descriptions reveal $k$ symbols from each of the $y_j, \ j \in \mathcal{N}$. Also, since the punctured coordinates are known to the decoder, it can construct a longer codeword from every partially received codeword by adding erasures in place of the punctured coordinates. The longer codewords can be treated as codewords from the original $(q-1,k)$ MDS code. The original MDS code can subsequently be decoded by applying an erasure decoding algorithm \cite[Ch. 9]{Wicker:coding} and all the $p_j$ vectors can be recovered. Mapping the $p_j$ vectors to their binary representations reveals the erased version $X_e^l$ of the original source string $X^l$. Output $\{(X_1,\ldots,X_m)\} \cup \{X_e^l \backslash (X_1,\ldots,X_m)\}$ as the reconstruction of $X^l$.

\emph{Analysis:} We now argue that the above scheme achieves the rate-distortion vector $(R,1-\frac{1}{n},1-\frac{2}{n},\ldots,1-\frac{k-1}{n},D_k,(\frac{n-k-1}{n-k})D_k,(\frac{n-k-2}{n-k})D_k,\ldots,(\frac{1}{n-k})D_k,0)$. For any source string $X^l$, every description (say the $i^{th}$ description) consists of $(X_i,y_{ji}: j \in \mathcal{N})$. $X_i$ consists of $lD_k/(n-k)$ bits. Now since $y_{ji}$ is an element of GF($q$), it can be represented by $m$ bits. Thus $(y_{ji}: j \in \mathcal{N})$ is a length-$n$ vector in GF($q$), and can be represented by $mn$ bits. Every description therefore consists of $mn + lD_k/(n-k)$ bits. Since the source string consists of $l=mnk(n-k)/(n(1-D_k)-k)$ source symbols, every description has rate $$ \frac{mn + lD_k/(n-k)}{l} = \frac{1-D_k}{k}=R.$$ Moreover, every description received at the decoder reveals $lD_k/(n-k)$ bits via $X_i$, and exactly one symbol from $k$ of the $p_j, \ j \in \mathcal{N}$. Each of these $k$ symbols is an element of GF($q$) and can be represented by $m$ bits. Thus every description reveals $lD_k/(n-k)+ mk$ bits to the decoder. (We note that the bits revealed by any two descriptions are disjoint. The uncoded bits $X_a$ and $X_b$ are disjoint by definition for any two descriptions $a$ and $b$. Now suppose descriptions $a$ and $b$ revealed the same symbol from some $p_j$. Then $y_{ja} = p_j\textbf{G}_{ja} = p_j\textbf{G}_{jb} = y_{jb}$, which implies $a = b$.) Thus if $c < k$ descriptions are received, the decoder can reconstruct $c(lD_k/(n-k)+mk)$ bits of the original source sequence. Thus
\begin{align*}
D_c &= 1 - \frac{c(\frac{lD_k}{n-k}+mk)}{l} \\
&= 1-\frac{cD_k}{n-k}-\frac{cn(1-D_k)-ck}{n(n-k)} \\
&= 1-\frac{c}{n}.
\end{align*}

If $c \ge k$ descriptions are received, say descriptions $1,\ldots,m$, then $(X_1,\ldots,X_m)$ reveal $clD_k/(n-k)$ bits. Moreover, the erased version of the source sequence, $X_e^l$, can be reconstructed by applying the MDS erasure decoding algorithm. The bits revealed by $(X_1,\ldots,X_m)$ are disjoint from the bits revealed by $X_e^l$. The total number of bits revealed, therefore, is $clD_k/(n-k)+mnk$. Thus
\begin{align*}
D_c &= 1-\frac{c\frac{lD_k}{n-k}+mnk}{l} \\
&= 1-\frac{cD_k}{n-k}-\frac{n(1-D_k)-k}{n-k} \\
&= \left(\frac{n-c}{n-k}\right)D_k.
\end{align*}
Thus $(R,1-\frac{1}{n},1-\frac{2}{n},\ldots,1-\frac{k-1}{n},D_k,(\frac{n-k-1}{n-k})D_k,(\frac{n-k-2}{n-k})D_k,\ldots,(\frac{1}{n-k})D_k,0) \in \mathcal{RD}_{worst}$.
\end{proof}

Figure~\ref{fig:results} depicts how the achievable distortion varies with the number of descriptions received at the decoder when $D_k = 0$.

\begin{figure}[htp]
\centering
  \includegraphics[scale=0.25]{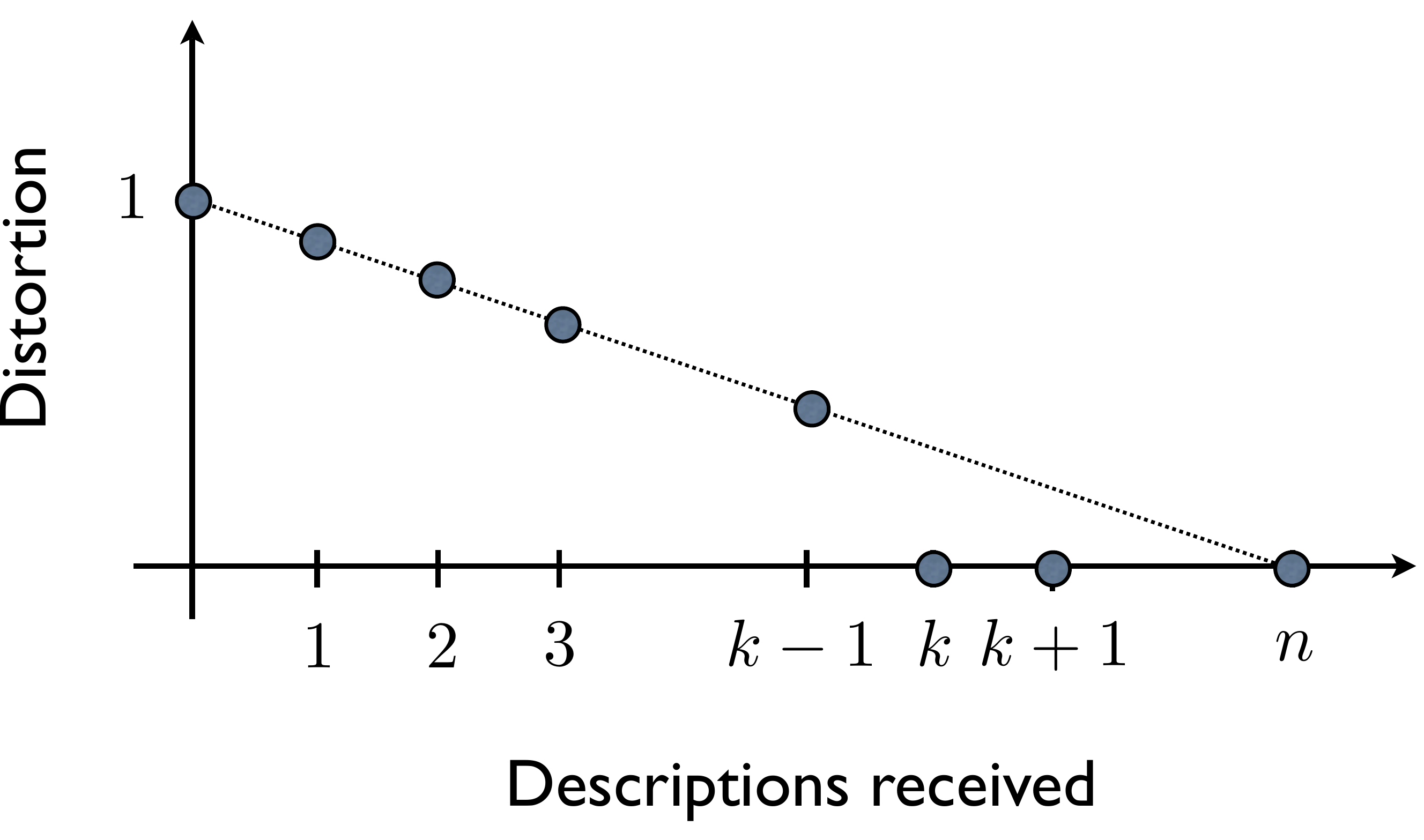}
\caption{The achievable distortion region for $D_k = 0$. The achievable distortion decreases linearly with the number of descriptions received up to $k-1$ descriptions, and drops abruptly to zero upon reception of $k$ or more descriptions.}\label{fig:results}
\end{figure}

\subsection{Optimality Results}
We now present optimality results for the MDS coding scheme described in the previous subsection. These optimality results are stronger than those for average-case distortion and yield a more complete characterization of the achievable distortion region. Since we are dealing with worst-case distortion constraints, the following results hold for any source distribution.

\begin{theorem}
\label{thm:conv_w-case_pt-2-pt}
For any $n$ and $k$, if $D_k \ge 1-\frac{k}{n}$ and rational\footnote{For this theorem and subsequent theorems in this subsection, we consider rational values for $D_k$ since any code over a finite blocklength can yield only rational distortions.}, then for any $(R,D_1,\ldots,D_k,\ldots,D_n) \in \mathcal{RD}_{worst}$, $D_m \ge 1-mR$ for all $m \in \mathcal{N}$.
\end{theorem}
\begin{proof}
Let $D_k \ge 1-\frac{k}{n}$. If a code achieves a certain distortion under worst-case distortion, then it will achieve that distortion under average-case distortion as well. The result therefore follows from the first part of Theorem~\ref{thm:conv_avg-case}.
\end{proof}

The following lemma is integral to the proofs of our optimality results for worst-case distortion.

\begin{definition}
Let $X^l$ be a random vector taking values in $\mathcal{X}^l$. An \textit{erased version} of $X^l$ is a random vector $\tilde{X}^l$, taking values in $\mathcal{\hat{X}}^l$, such that $\nexists \ t \in \{1,\ldots,l\}$ such that $\tilde{X}_t = +$ and $X_t = -$ or  $\tilde{X}_t = -$ and $X_t = +$.
\end{definition}

\begin{lemma}
\label{lemma:cases}
Let $X_1^l(X),X_2^l(X),\ldots,X_n^l(X)$ be erased versions of the source string $X^l \in \mathcal{X}^l$. Suppose $X^l$ is \emph{i.i.d.} uniform over $\mathcal{X}^l$. If for all $t \in \{1,\ldots,l\}$, $I(X_{it}(X);X_{jt}(X))= 0$ $\forall \ i,j \in \mathcal{N}$, $i \neq j$, then
\begin{align*}
\max_{x^l \in \mathcal{X}^l} \sum_{i=1}^n \left[\frac{1}{l}\sum_{t=1}^l d(x_t,X_{it}(x))\right] \ge n-1.
\end{align*}
\end{lemma}
\begin{proof}
See Appendix~\ref{app:lemma_cases}.
\end{proof}

The following theorem proves that the MDS coding scheme is optimal for all $n$ and $k$ when a single-message is received at the decoder.

\begin{theorem}
\label{thm:conv_w-case_D_1}
For any $n$ and $k$, if $D_k < 1-\frac{k}{n}$ and rational, then for any $(R,D_1,\ldots,D_k,\ldots,D_n) \in \mathcal{RD}_{worst}$, $D_1 \ge 1-\frac{1}{n}$.
\end{theorem}
\begin{proof}
See Appendix~\ref{app:conv_w-case_D_1}.
\end{proof}

The following theorem shows that the MDS coding scheme is Pareto optimal in the distortions $D_1,\ldots,D_{k-1}$.

\begin{theorem}
\label{thm:pareto_optimality}
For any $n$ and $k$, $(R,1-\frac{1}{n},1-\frac{2}{n},\ldots,1-\frac{k-1}{n},D_k,(\frac{n-k-1}{n-k})D_k,(\frac{n-k-2}{n-k})D_k,\ldots,(\frac{1}{n-k})D_k,0)$ is \emph{Pareto optimal} in $D_1,\ldots,D_{k-1}$, i.e., there does not exist $(R',D'_1,\ldots,D'_n) \in \mathcal{RD}_{worst}$ such that either $R' < R$, or $R' \le R$, $D'_i \le 1-\frac{i}{n}$ for all $1 \le i \le k-1$ and $D'_j < 1-\frac{j}{n}$ for at least one $j$, $1 \le j \le k-1$.
\end{theorem}
\begin{proof}
See Appendix~\ref{app:pareto_optimality}.
\end{proof}

The following theorem shows that for certain values of $m$, $n$ and $k$, the MDS coding scheme is optimal when $m$ messages are received.

\begin{theorem}
\label{thm:other_wc_distortions}
For any $n$ and $k$, if $m \le \frac{k}{2}$ and $m|n$ ($m$ divides $n$), then for any $(R,D_1,\ldots,D_k,\ldots,D_n) \in \mathcal{RD}_{worst}$, $D_m \ge 1-\frac{m}{n}$.
\end{theorem}
\begin{proof}
See Appendix~\ref{app:other_wc_distortions}.
\end{proof}

\section{A General Multiple Descriptions Architecture}
\label{sec:gauss_md}
The schemes described in this paper provide a substrate that can be used to construct no-excess-rate multiple descriptions codes for a general source using only a point-to-point rate-distortion code for that source. We illustrate this idea for a Gaussian source, where the resulting scheme is optimal in a certain sense. The extension to arbitrary sources should be clear from the proof. Suppose that $(\textbf{X}_t)_{t=1}^\infty$ is a memoryless Gaussian process, where $\textbf{X}_t$ is a vector of length $N$ and has a marginal distribution $\mathcal{N}(0,\textbf{K}_x)$. The distortion for a source-reconstruction pair $(\textbf{X}^l,\hat{\textbf{X}}^l)$ is measured as $\textbf{E}\left[\frac{1}{l}\sum_{t=1}^l (\textbf{X}_t-\hat{\textbf{X}}_t)(\textbf{X}_t-\hat{\textbf{X}}_t)^T\right].$ We compare distortions in the positive definite sense, i.e., $\textbf{D}_A \succcurlyeq \textbf{D}_B$ iff $\textbf{D}_A - \textbf{D}_B \succcurlyeq \textbf{0}$.

\begin{definition}
The rate-distortion vector $(R,\textbf{D}_1,\ldots,\textbf{D}_n)$ is \emph{achievable} if for some $l$ there exist encoders $f_i^{(l)}: \mathbb{R}^{N\times l} \rightarrow \{1,\ldots,M_i^{(l)}\}$, $i\in \mathcal{N}$ and decoders $g_\mathcal{K}^{(l)}: \prod_{k \in \mathcal{K}} \{1,\ldots,M_k^{(l)}\}  \rightarrow \mathbb{R}^{N\times l}$, $\mathcal{K}\subseteq \mathcal{N}$, $\mathcal{K} \neq \emptyset$, such that
\begin{align*}
R &\ge \frac{1}{l}\log M_i^{(l)} \ \textrm{ for all $i$, and} \\
\textbf{D}_k &\succcurlyeq \textbf{E}\left[\frac{1}{l}\sum_{t=1}^l (\textbf{X}_t-\hat{\textbf{X}}_{\mathcal{K},t})(\textbf{X}_t-\hat{\textbf{X}}_{\mathcal{K},t})^T\right] \ \text{for all $\mathcal{K} \subseteq \mathcal{N}$, $|\mathcal{K}| = k$},
\end{align*}
where $\hat{\textbf{X}}_{\mathcal{K}}^l = \textbf{E}[\textbf{X}^l|f_i^{(l)}(\textbf{X}^l), i \in \mathcal{K}].$
\end{definition}
We use $\mathcal{RD}_{gauss}$ to denote the set of achievable rate-distortion vectors and $\overline{\mathcal{RD}}_{gauss}$ to denote its closure. We consider symmetric descriptions, i.e., each description has the same rate $R_g$ and the distortion constraint depends only on the number of descriptions received. We consider the case where there is \emph{no excess rate} for every $k$ out of $n$ descriptions, i.e., $kR_g = R(\textbf{D}_k)$, where $R(\cdot)$ is the Shannon rate-distortion function and
\begin{align*}
R(\textbf{D}_k) &= \min_{\tilde{\textbf{D}}} \frac{1}{2}\log \frac{|\textbf{K}_x|}{|\tilde{\textbf{D}}|} \\
&\phantom{=} \text{s.t. } \tilde{\textbf{D}} \preccurlyeq \textbf{D}_k \text{ and } \\
&\phantom{=\text{s.t. }} \tilde{\textbf{D}} \preccurlyeq \textbf{K}_x.
\end{align*}
Thus $R_g = \frac{1}{k}R(\textbf{D}_k)$ bits/symbol.

\begin{theorem}
$\left(R_g,\frac{\textbf{D}_k + (n-1)\textbf{K}_x}{n},\frac{2\textbf{D}_k + (n-2)\textbf{K}_x}{n},\ldots,\frac{(k-1)\textbf{D}_k + (n-k+1)\textbf{K}_x}{n},\textbf{D}_k,\ldots,\textbf{D}_k\right) \in \overline{\mathcal{RD}}_{gauss}$.
\end{theorem}
\begin{proof}
Fix $\textbf{D}_k$ and consider an integer $l$. We know from rate-distortion theory that there exists an integer $l' \ge lR(\textbf{D}_k)$ such that any source sequence $\textbf{X}^l$ of $l$ symbols can be compressed to a sequence $Y^{l'}$ consisting of $l'$ bits and then reproduced from $Y^{l'}$ with distortion $\preccurlyeq \textbf{D}_k + \epsilon \mathbf{I}$ for $l$ sufficiently large. Chose now a blocklength $nl$. The $nl$ source symbols can be compressed into a binary sequence $Y^{nl'}$ taking values in $\mathcal{X}$, which can then be transmitted to the decoder over the $n$ channels using the achievability scheme proposed in Section~\ref{sec:worst-case-achievability}. Thus every description contains $l'$ uncoded bits of $Y^{nl'}$. In particular, the decoder should be able to completely reconstruct $Y^{nl'}$ upon reception of any $k$ descriptions, i.e, there is no distortion for every $k$ out of $n$ descriptions (this corresponds to a special case of Theorem~\ref{thm:achievability_avg-case} with $D_k = 0$). Thus every set of $k$ descriptions must reveal $nl'$ bits, and therefore the rate of a single description is $\tilde{R} = nl'/knl = l'/kl$ bits per symbol of $\textbf{X}^l$. Moreover, since every description contains $l'$ uncoded bits, the decoder can reconstruct $ml'$ bits of $Y^{l'}$ upon reception of $m < k$ descriptions.

We now argue that $\left(R_g,\frac{\textbf{D}_k + (n-1)\textbf{K}_x}{n},\ldots,\frac{(k-1)\textbf{D}_k + (n-k+1)\textbf{K}_x}{n},\textbf{D}_k,\ldots,\textbf{D}_k\right) \in \overline{\mathcal{RD}}_{gauss}$. The rate of every description is $\tilde{R} = l'/kl$. Moreover, any $m < k$ descriptions reveal $ml'$ bits of $Y^{nl'}$. It follows from a time-sharing argument that upon receptions of $m < k$ descriptions, the decoder can reconstruct $\textbf{X}^{nl}$ with distortion $\frac{m(\textbf{D}_k + \epsilon\mathbf{I}) + (n-m)\textbf{K}_x}{n}$. When $k$ or more descriptions are received, the decoder is able to reconstruct $Y^{nl'}$ completely and can reconstruct $\textbf{X}^{nl}$ with distortion less than $ \preccurlyeq \textbf{D}_k + \epsilon \mathbf{I}$. Now let $l \rightarrow \infty$. Then we can let $l' \rightarrow \infty$ such that $\frac{l'}{l} \rightarrow R(\textbf{D}_k)$ and $\epsilon \rightarrow 0$. Thus $\tilde{R} = \frac{l'}{kl} \rightarrow \frac{1}{k}R(\textbf{D}_k) = R_g$, and so
$\left(R_g,\frac{\textbf{D}_k + (n-1)\textbf{K}_x}{n},\ldots,\frac{(k-1)\textbf{D}_k + (n-k+1)\textbf{K}_x}{n},\textbf{D}_k,\ldots,\textbf{D}_k\right) \in \overline{\mathcal{RD}}_{gauss}$.
\end{proof}

Next, we show that, for the special case of symmetric scalar Gaussian multiple descriptions with two levels of receivers (where one receiver reconstructs the source from any $k$ out of $n$ descriptions with distortion $\textbf{D}_k$ and the second receiver reconstruct the source from all $n$ description with distortion $\textbf{D}_n$), and no excess rate for the second receiver, the aforementioned scheme achieves the optimal $\textbf{D}_k$. It has been shown by Wang and Viswanath \cite[Theorem 1]{wang_visw:two-levels} that given distortion constraints $\textbf{D}_k$ and $\textbf{D}_n$, the symmetric multiple description rate for an \emph{i.i.d.} vector Gaussian source with mean $\textbf{0}$ and convariance $\textbf{K}_x$ is $$\hat{R} = \sup_{\textbf{K}_z \succ \textbf{0}} \quad \frac{1}{2} \log \left(\frac{|\textbf{K}_x|^{\frac{1}{n}} |\textbf{K}_x + \textbf{K}_z|^{\frac{n-k}{kn}} |\textbf{D}_n + \textbf{K}_z|^{\frac{1}{n}}}{|\textbf{D}_n|^{\frac{1}{n}} |\textbf{D}_k + \textbf{K}_z|^{\frac{1}{k}}}\right).$$ Thus the sum rate of the $n$ descriptions is
\begin{align}
\label{eqn:vw_sum-rate}
n\hat{R} &=  \sup_{\textbf{K}_z \succ \textbf{0}} \quad \frac{1}{2} \log \left(\frac{|\textbf{K}_x| |\textbf{K}_x + \textbf{K}_z|^{\frac{n-k}{k}} |\textbf{D}_n + \textbf{K}_z|}{|\textbf{D}_n| |\textbf{D}_k + \textbf{K}_z|^{\frac{n}{k}}}\right).
\end{align}

\begin{theorem}
For scalar Gaussian multiple descriptions (\emph{i.i.d.} $\mathcal{N}(0,\sigma_x^2)$ Gaussian source) with two levels of receivers (distortion constraints $D_k$ and $D_n$, respectively) and no excess rate for the second receiver, $D_k \ge \frac{k}{n}D_n + \frac{n-k}{n}\sigma_x^2$.
\end{theorem}
\begin{IEEEproof}
Assume WLOG that $\sigma_x^2 = 1$. Reducing \eqref{eqn:vw_sum-rate} to the scalar case and using the no excess rate condition gives
\begin{align*}
\frac{1}{2} \log \left(\frac{1}{D_n}\right) &=  \sup_{\lambda > 0} \quad \frac{1}{2} \log \left(\frac{1}{D_n} \cdot \frac{(1 + \lambda)^{\frac{n-k}{k}} (D_n + \lambda)}{(D_k + \lambda)^{\frac{n}{k}}}\right),
\end{align*}
which implies
\begin{align*}
0 = \sup_{\lambda > 0} \quad \frac{1}{2} \log \left(\frac{(1 + \lambda)^{\frac{n-k}{k}} (D_n + \lambda)}{(D_k + \lambda)^{\frac{n}{k}}}\right).
\end{align*}
Define $f(\lambda) = \frac{(1+\lambda)^{\frac{n}{k}-1}(D_n+\lambda)}{(D_k+\lambda)^{\frac{n}{k}}}.$ Then
\begin{align*}
0 &= \sup_{\lambda > 0} \log_e f(\lambda) \\
&= \sup_{\lambda > 0} \left(\frac{n}{k}-1\right) \log_e (1 + \lambda) + \log_e (D_n + \lambda) - \frac{n}{k} \log_e (D_k + \lambda) \\
&= \sup_{\lambda > 0} \log_e \frac{D_n + \lambda}{1 + \lambda} + \frac{n}{k} \log_e \frac{1 + \lambda}{D_k + \lambda} \\
&= \sup_{\lambda > 0} \log_e \left(1 + \frac{D_n - 1}{1 + \lambda}\right) + \frac{n}{k} \log_e \left(1 + \frac{1 - D_k}{D_k + \lambda}\right).
\end{align*}

Define $$g(\lambda) = \frac{(\frac{D_n - 1}{1+\lambda})^2}{2(1 - |\frac{D_n - 1}{1+\lambda}|)^2} + \frac{(\frac{1 - D_k}{D_k+\lambda})^2}{2(1 - |\frac{1 - D_k}{D_k+\lambda}|)^2}.$$ Using the fact that
\begin{align*}
\log_e (1 + x) \ge x - \frac{x^2}{2(1 - |x|)^2} \ \text{for $|x| < 1$}
\end{align*}
we obtain
\begin{align*}
0 &\ge \sup_{\lambda > 0} \left(\frac{D_n - 1}{1 + \lambda} + \frac{n}{k} \left(\frac{1 - D_k}{D_k + \lambda}\right) - g(\lambda)\right) \\
\frac{1-D_n}{1 + \lambda} &\ge \frac{n}{k} \left(\frac{1 - D_k}{D_k + \lambda}\right) - g(\lambda) \\
\frac{D_k + \lambda}{1 + \lambda} &\ge \frac{n}{k} \left(\frac{1 - D_k}{1 - D_n}\right) - \frac{D_k + \lambda}{1 - D_n}g(\lambda).
\end{align*}
Now let $\lambda \rightarrow \infty$. Then $\frac{D_k + \lambda}{1 - D_n}g(\lambda) \rightarrow 0$ and $\frac{D_k + \lambda}{1 + \lambda} \rightarrow 1$. We thus have
\begin{align*}
1 &\ge \frac{n}{k} \left(\frac{1 - D_k}{1 - D_n}\right) \\
\Rightarrow D_k &\ge \frac{k}{n}D_n + \frac{n-k}{n}.
\end{align*}
\end{IEEEproof}

\section{Decentralized Encoding}
\label{sec:dec_encoding}
In this section we characterize the optimal distortion tradeoff for the robust binary erasure CEO problem. The robust binary erasure CEO problem is a generalization of the multiple descriptions problem studied earlier in that the encoders observe an erased version of the source instead of the source itself. In particular, let $Y_i = N_i \cdot X, \ i \in \mathcal{N}$, where $X \in \mathcal{X}$ and $N_1,\ldots,N_n$ are \emph{i.i.d.} Bernoulli with $0 < \Pr(N_i = 0) = p < 1$. Thus the $Y_i$ take values in $\mathcal{\hat{X}} = \{+,-,0\}$. A \emph{encoder} is a function $f_i:\mathcal{\hat{X}}^l \rightarrow \left\{1,\ldots,M_i^l\right\}, \ i \in \mathcal{N}$. A \emph{decoder} is a function $g_\mathcal{K}: \prod_{k \in \mathcal{K}} \left\{1,\ldots,M_k^l\right\} \rightarrow \hat{\mathcal{X}}^l$, where $\mathcal{K} \subseteq \mathcal{N}$ is the set of messages received. There are $n$ encoders. Encoder $f_i, \ i \in \mathcal{N}$, observes $Y_i^l$  and transmits an encoded version of it over channel $i$. The receiver either receives this description without errors or is not able to receive it at all. Excluding the case where none of the messages is received, the receiver may receive $2^n - 1$ different combinations of the $n$ messages. Thus it can be represented by the $2^n-1$ decoding functions $g_\mathcal{K}, \ \mathcal{K} \subseteq \mathcal{N}$, $\mathcal{K} \neq \emptyset$. Based on the set of received messages $\mathcal{K}$, the receiver employs the corresponding decoding function to output a reconstruction $\hat{X}_\mathcal{K}^l$ of the original source string $X^l$ subject to a distortion constraint. We consider symmetric rates, i.e., each message has the same rate $R$ and the distortion constraint depends only on the number of messages received.

We measure the fidelity of the reconstruction using a family of distortion measures, $\{d^\lambda\}_{\lambda > 0}$, where
\begin{equation*}
d^\lambda(x,\hat{x}) =
\begin{cases}
0 & \text{if $\hat{x} = x$} \\
1 & \text{if $\hat{x} = 0$} \\
\lambda & \text{otherwise}.
\end{cases}
\end{equation*}
We are particularly interested in the large-$\lambda$ limit. In this regime, $d^\lambda$ approximates the erasure distortion measure. We use this family of finite distortion measures because an infinite distortion measure is too harsh for this setup: it does not allow decoding errors at all, even those that have negligible probability.

\begin{definition}
The rate-distortion vector $(R,D_1,D_2,\ldots,D_n)$ is \emph{achievable} if there exists a block length $l$ for which there exist encoders
$f_i, \ i \in \mathcal{N}$, and decoders $g_\mathcal{K}$, $\mathcal{K}\subseteq \mathcal{N}$, $\mathcal{K} \neq \emptyset$ such that
\begin{equation}
\label{ceo_constraints}
\begin{split}
R & \ge \frac{1}{l} \log M_i^{(l)} \ \text{for all $i \in \mathcal{N}$, and} \\
D_k & \ge E\left[\frac{1}{l} \sum_{t = 1}^l
d^\lambda(X_t,\hat{X}_{\mathcal{K}t,})\right] \
\text{for all subsets of messages $\mathcal{K}, |\mathcal{K}|=k$}.
\end{split}
\end{equation}
Let $\mathcal{RD}_{CEO}(\lambda)$ denote the set of achievable rate-distortion vectors. Define
\begin{align*}
\mathcal{RD}_{CEO} = \bigcap_{\lambda \ge 1}^\infty \mathcal{RD}_{CEO}(\lambda).
\end{align*}
\end{definition}
We use $\overline{\mathcal{RD}}_{CEO}$ to denote the closure of $\mathcal{RD}_{CEO}$. Our main result is the characterization of the optimal distortion tradeoff for an arbitrary code with respect to the number of messages received. We show that if a code comes arbitrarily close to achieving the minimum achievable distortion $D_k$ upon reception of $k$ messages, then the distortion it can achieve upon reception of $\ell$ messages cannot be lower than $D_k^{\ell/k}$. Achievability can be shown by using a random binning scheme based on $(n,k)$ source-channel erasure codes, proposed in \cite{pradpuriramc04}. The result therefore proves that $(n,k)$ source-channel erasure codes are optimal for this setup. Informally, the scheme involves constructing a codebook $\mathcal{C}_i$ for $Y_i$ at encoder $f_i$ and then binning all the codewords independently and uniformly. Encoder $f_i$ observes $Y_i^l$ and then sends the bin index of the corresponding codeword to the decoder. Upon receiving the messages, the decoder searches the corresponding bins and outputs a reconstruction of the source sequence based on the bits revealed by the codewords. If none of the decoded codewords reveal a particular source bit, then the decoder just outputs an erasure in place of that bit. It can be verified that, for this scheme, if the distortion upon reception of any $k$ messages is $D_k$, then the distortion upon reception of any $\ell$ messages is $D_k^{\ell/k}$. The intuition is that if $q$ is the probability that a particular bit is not revealed by a particular message, then the chance that $k$ messages will not reveal that bit is $q^k$, and the chance that $\ell$ messages will not reveal that bit is $q^\ell = (q^k)^{\ell/k}$.

Before proving the converse for this problem, we will state and prove an outer bound on the rate region of the multi-terminal source coding problem in the next subsection. We will then use this bound to prove our result in Section~\ref{sec:CEO_sub}.

\subsection{Outer Bound on the Rate Region of the Multi-terminal Source Coding Problem}
\label{sec:out_bound}
The term ``multi-terminal source coding" typically refers to the problem of reconstructing two correlated, separately encoded sources, each subject to a distortion constraint. In this paper we use the term to refer to the more general model considered in \cite{wagnervenkat05}: we have an arbitrary number of sources $Y_1,\ldots,Y_n$, with $Y_i$ taking values in the set $\mathcal{Y}_i$, encoders $f_i, \ i \in \mathcal{N}$, a hidden source $Y_0$ which is not directly observed by any encoder or the decoder, and a side information source $Y_{n+1}$, taking values in the set $\mathcal{Y}_{n+1}$, which is observed by the decoder but not by any encoder. In particular, $\{Y_{0,t},Y_{1,t},\ldots,Y_{n,t},Y_{n+1,t}\}_{t=1}^\infty$ is a vector-valued, finite-alphabet and memoryless source. Encoder $f_i$ observes a length-$l$ sequence of $Y_i$ and transmits a message to the decoder based on the mapping $$f_i^{(l)}:\mathcal{Y}_i^l \rightarrow \left\{1,\ldots,M_i^{(l)}\right\}.$$ We allow the decoder to reconstruct arbitrary functions of the sources $V_1,\ldots,V_J$ (with $V_j, \ j = 1,\ldots,J$ taking values in the set $\mathcal{V}_j$) instead of, or in addition to, the sources themselves. We also allow the decoder to reconstruct $V_1,\ldots,V_J$ from subsets of messages $f_\mathcal{K} = \{f_k^{(l)}, \ k \in \mathcal{K}\}$, where $\mathcal{K} \subset \mathcal{N}, \mathcal{K} \ne \emptyset$. The decoder thus uses the mappings $$\left(g_\mathcal{K}^j\right)^{(l)}:\mathcal{Y}_{n+1}^l \times \prod_{k \in \mathcal{K}} \left\{1,\ldots,M_k^{(l)}\right\} \rightarrow \mathcal{V}_j^l, \ \text{for } \mathcal{K} \subset \mathcal{N}, \mathcal{K} \ne \emptyset, j = 1,\ldots,J.$$ We thus have $J$ distortion measures $$d_j:\prod_{i=0}^{n+1} \mathcal{Y}_i \times \mathcal{V}_j \rightarrow \mathbb{R}^+.$$ For every $j = 1,\ldots,J$, we impose a common distortion constraint for all size-$k$ subset of messages used to reconstruct $V_j$. More precisely, for every $j = 1,\ldots,J$, all ${n \choose k}$ subsets of messages of size $k$, when used to reconstruct $V_j$, must satisfy a single distortion constraint. Thus there are $nJ$ distortion constraints in total. We will use the following notation and definitions from \cite{wagnervenkat05}. Let $\mathbf{Y}_\mathcal{K}$ denote $(Y_k)_{k \in \mathcal{K}}$, and $Y_{i^c}$ denote $Y_{\{i\}^c}$. Moreover, $Y_{i,a:b}$ denotes $\{Y_{i,a},Y_{i,a+1},\ldots,Y_{i,b}\}$.

\begin{definition}
The rate-distortion vector $$(\mathbf{R},\mathbf{D}) = (R_1,R_2,\ldots,R_n,D_{1,1},D_{2,1},\ldots,D_{n,1},D_{1,2},\ldots,D_{n,2},\ldots,D_{1,J},\ldots,D_{n,J})$$ is achievable if for some $l$ there exist encoders $f_i^{(l)}, \ i \in \mathcal{N}$, and decoders $(g_\mathcal{K}^j)^l, \ \mathcal{K} \subset \mathcal{N}, \mathcal{K} \ne \emptyset, j = 1,\ldots,J$, such that
\begin{align}
\label{constraints}
\begin{split}
R_i &\ge \frac{1}{l} \log M_i^{(l)}, \ i \in \mathcal{N}, \text{ and } \\
D_{k,j} &\ge \max_{\mathcal{K}:|\mathcal{K}|=k} \mathbf{E}\left[\frac{1}{l} \sum_{t=1}^l d_j(Y_{0,t},\mathbf{Y}_{\mathcal{K},t},Y_{n+1,t},V_{j,t})\right] \ \text{for } j = 1,\ldots,J.
\end{split}
\end{align}
\end{definition}
As in \cite{wagnervenkat05}, we use $\mathcal{RD}_\star$ to denote the set of achievable rate-distortion vectors and $\overline{\mathcal{RD}_\star}$ to denote its closure. We use the following definitions from \cite{wagnervenkat05}.

\begin{definition}
Let $Y_0,Y_1,\ldots,Y_{n+1}$ be generic random variables with the distribution of the source at a single time. Let $\Gamma_o$ denote the set of finite-alphabet random variables $\gamma = (U_1,\ldots,U_n,V_1,\ldots,V_j,W,T)$ satisfying
\begin{enumerate}
\item[(i)]
$(W,T)$ is independent of $(Y_0,\mathbf{Y}_\mathcal{N},Y_{n+1})$,
\item[(ii)]
$U_i \leftrightarrow (Y_i,W,T) \leftrightarrow (Y_0,\mathbf{Y}_{i^c},Y_{n+1},\mathbf{U}_{i^c}$), shorthand for
``$U_i$, $(Y_i,W,T)$ and $(Y_0,\mathbf{Y}_{i^c},Y_{n+1},\mathbf{U}_{i^c})$ form a Markov chain in this order'', for all $i \in \mathcal{N}$, and
\item[(iii)]
$(Y_0,\mathbf{Y}_\mathcal{N},W) \leftrightarrow (\mathbf{U}_\mathcal{N},Y_{n+1},T) \leftrightarrow (V_1,\ldots,V_j)$.
\end{enumerate}
\end{definition}

\begin{definition}
Let $\psi$ denote the set of finite-alphabet random variables $Z$ with the property that $Y_1,\ldots,Y_n$ are conditionally independent given $(Z,Y_{n+1})$.
\end{definition}
There are many ways of coupling a given $Z \in \psi$ and $\gamma \in \Gamma_o$ to the source. In this paper, we shall only consider the \emph{Markov coupling} for which $Z \leftrightarrow (Y_0,\mathbf{Y}_\mathcal{N},Y_{n+1}) \leftrightarrow \gamma$. We now state our outer bound.

\begin{definition}
\label{Me:defn} Let
\begin{align*}
\mathcal{RD}_o(Z,\gamma) & = \Bigg\{(\mathbf{R},\mathbf{D}) :
\sum_{i \in \mathcal{K}} R_i \ge \max \Big(I(Z;\mathbf{U}_\mathcal{K}|Y_{n+1},T),I(Z;\mathbf{U}_\mathcal{K}|\mathbf{U}_{\mathcal{K}^c},Y_{n+1},T)\Big) \\
& \phantom{= \Bigg\{}   \mbox{} + \sum_{i \in \mathcal{K}} I(Y_i;U_i|Z,Y_{n+1},W,T) \ \text{for all $\mathcal{K} \subseteq \mathcal{N}$}, \\
& \phantom{= \Bigg\{} \ \text{and} \ D_{k,j} \ge \max_{\mathcal{K}:|\mathcal{K}|=k} \mathbf{E}[d_j(Y_0,\mathbf{Y}_\mathcal{K},Y_{n+1},V_j)] \ \text{for } j = 1,\ldots,J \Bigg\}.
\end{align*}
Then define
\begin{equation*}
\mathcal{RD}_o = \bigcap_{Z \in \psi} \bigcup_{\gamma \in \Gamma_o}
          \mathcal{RD}_o(Z,\gamma).
\end{equation*}
\end{definition}

\begin{theorem}
\label{main} $\mathcal{RD}_\star \subseteq \mathcal{RD}_o$.
\end{theorem}
\begin{proof}
See Appendix~\ref{app:outer_bound}.
\end{proof}

The proposed bound differs in two ways from the bound in \cite{wagnervenkat05} as follows. Whereas the bound in \cite{wagnervenkat05} lower bounds the sum rate of a subset $\mathcal{K}$ of messages by $I(Z;\mathbf{U}_\mathcal{K}|\mathbf{U}_{\mathcal{K}^c},Y_{n+1},T)$, the proposed bound potentially improves upon it by taking the maximum of $I(Z;\mathbf{U}_\mathcal{K}|\mathbf{U}_{\mathcal{K}^c},Y_{n+1},T)$ and $I(Z;\mathbf{U}_\mathcal{K}|Y_{n+1},T)$. Moreover, the proposed bound imposes distortion constraints for source reproductions based on subsets of messages, rather than only for reproductions based on all of the messages. These improvements were needed in order to use the bound to prove our converse result for the robust CEO problem: the robust CEO problem requires the decoder to be able to reconstruct the source sequence from a subset $\mathcal{K}$ of the encoded messages, subject to a distortion constraint, without having any knowledge about the messages in $\mathcal{K}^c$. The outer bound in \cite{wagnervenkat05} cannot be applied to this problem, since, as mentioned earlier, it lower bounds the sum rate of the subset of messages $\mathcal{K}$ by $I(Z;\mathbf{U}_\mathcal{K}|\mathbf{U}_{\mathcal{K}^c},Y_{n+1},T)$ which involves conditioning on the messages in $\mathcal{K}^c$.

Although we apply our improved outer bound to the robust binary erasure CEO problem in this paper, we believe that it could potentially be useful for other instances of the multi-terminal source coding problem.

\subsection{Optimal Distortion Tradeoffs for Robust CEO}
\label{sec:CEO_sub}

As defined earlier, the robust binary erasure CEO problem is an instance of the general multi-terminal source coding problem in which the hidden source $Y_0$ takes values in $\mathcal{X} = \{+,-\}$. There is no side information $Y_{n+1}$ and the decoder is interested in reproducing an estimate $V_1$ of the hidden source $Y_0$ only. In order to be consistent with the notation used in the beginning of this section, we shall henceforth use $X$ instead of $Y_0$ and $\hat{X}$ instead of $V_1$. Here $\hat{X}$ takes values in $\{+,-,0\}$. We begin with a few lemmas.

Let $g(\cdot)$ denote the function on $[p,\infty)$ defined by
\begin{equation*}
g(x) = \begin{cases}
h(x) - (1-p) h(\frac{x - p}{1 - p}) & p \le x \le 1 \\
0 & x > 1.
\end{cases}
\end{equation*}
The following corollary and lemma, which we state without proof, are from \cite{wagnervenkat05}.

\begin{corollary}\emph{\textbf{\cite[Corollary 1]{wagnervenkat05}}}
\label{convexcor}
The function $g(y^{1/n})$ is non-increasing and convex in $y$ on $[p^n,\infty)$.
\end{corollary}

\begin{lemma}\emph{\textbf{\cite[Lemma 6]{wagnervenkat05}}}
\label{continuity}
Suppose $p^n \le D$ and $(\mathbf{U},\hat{X})$ is such that
\begin{enumerate}
\item[(\emph{i})] $\mathbf{E}[d^\lambda(X,\hat{X})] \le D$,
\item[(\emph{ii})] $U_i \leftrightarrow Y_i \leftrightarrow (X,\mathbf{Y}_{i^c},
   \mathbf{U}_{i^c})$ for all $i \in \mathcal{N}$, and
\item[(\emph{iii})] $(X,\mathbf{Y}) \leftrightarrow \mathbf{U} \leftrightarrow \hat{X}$.
\end{enumerate}
If
\begin{equation*}
\frac{32n}{p(1-p)} \left(\frac{2D}{\lambda}\right)^{1/n} \le \delta
  \le \frac{1}{2},
\end{equation*}
then
\begin{equation*}
\frac{1}{n} \sum_{i = 1}^n I(Y_i;U_i|X) \ge
  g\left((D+ \delta)^{1/n}\right)
   + 2\delta \log\frac{\delta}{5}.
\end{equation*}
\end{lemma}

For the robust binary CEO problem, let $\hat{X}_\mathcal{K}^l$ be the source reconstruction when the subset $\mathcal{K}$ of messages is received at the receiver. We have the following lemma:

\begin{lemma}
\label{contrapositive} Suppose $p^\ell \le D$ and that $(\mathbf{U},X,\hat{X}_\mathcal{K},\mathbf{Y},W,T)$ for all $\mathcal{K}, |\mathcal{K}|= \ell$ is such that
\begin{enumerate}
\item[(i)] $(X,\mathbf{Y},\mathbf{U}_{\mathcal{K}^c},W) \leftrightarrow (\mathbf{U}_\mathcal{K},T) \leftrightarrow \hat{X}_\mathcal{K}$,
\item[(ii)] $U_i \leftrightarrow (Y_i,W,T) \leftrightarrow (X,\mathbf{Y}_{i^c},\mathbf{U}_{i^c})$ for all $i \in \mathcal{N}$, and
\item[(iii)] $\frac{1}{\ell} \sum_{i \in \mathcal{K}} I(Y_i;U_i|X, W, T) \le g(D^{1/\ell})$.
\end{enumerate}
Let $\tilde{D} = \max_{\mathcal{K}:\mathcal{K}=\ell} E[d^\lambda(X,\hat{X}_\mathcal{K})]$. For $\delta \in (0,1/2]$, if
\begin{equation*}
\lambda \ge \max \left[4\left(\frac{32\ell}{\delta p(1-p)}\right)^{2\ell}, \left(\frac{\tilde{D}}{\delta}\right)^2\right],
\end{equation*}
then
\begin{equation*}
\tilde{D} \ge D-\xi(\tilde{D},\delta)
\end{equation*}
for some continuous $\xi \ge 0$ satisfying $\xi(\tilde{D},0)=0$.
\end{lemma}
\begin{proof}
See Appendix~\ref{app:contrapositive}.
\end{proof}

We now prove our main result. Define
\begin{equation*}
\mathcal{R}_o(\mathbf{D},\lambda) = \inf\left\{R :
 (R,D_1,\ldots,D_n) \in \overline{\mathcal{RD}_o}(\lambda)\right\},
\end{equation*}
where $\overline{\mathcal{RD}_o}(\lambda)$ is the region given by Definition~\ref{Me:defn} when the distortion measure is $d^\lambda$.

It was shown in \cite[Section 3.2]{wagnervenkat05} that the sum rate of the binary erasure CEO problem with $n$ encoders, given a distortion constraint $D_n$, is\footnote{All logarithms and exponentiations in \cite{wagnervenkat05} have base $e$. Therefore the corresponding sum rate expression in \cite{wagnervenkat05} is $\sum_{i=1}^n R_i = (1-D_n)\log{2} + n \cdot g(D^{\frac{1}{n}}_n)$.}
\begin{displaymath}
\sum_{i=1}^n R_i = (1-D_n) + n \cdot g(D^{\frac{1}{n}}_n).
\end{displaymath}
It follows from this result, that for symmetric descriptions, if the distortion constraint for every subset of $k$ messages is $D_k$ and every message has rate $R$, then the sum rate for any $k$ descriptions is given by
\begin{align*}
kR = (1-D_k) + k \cdot g(D^{\frac{1}{k}}_k),
\end{align*}
which implies
\begin{equation}
\label{eqn:R_eq} R = \frac{(1-D_k)}{k} + g(D^{\frac{1}{k}}_k).
\end{equation}

\begin{theorem}
\label{thm:min_distortion} If $(R,D_1,\ldots,D_n) \in
\overline{\mathcal{RD}}_{CEO},$ and
\begin{align*}
D_k = \inf\Big\{D:(R,1,1,\ldots,1,D,1,\ldots,1) \in
\overline{\mathcal{RD}}_{CEO}\Big\},
\end{align*}
i.e.,
\begin{align*}
\label{eqn:R_eq} R = \frac{(1-D_k)}{k} + g(D^{\frac{1}{k}}_k),
\end{align*}
then
\begin{align*}
D_\ell \ge (D_k)^{\frac{\ell}{k}} \ \textrm{for all } \ell \ge k.
\end{align*}
\end{theorem}
\begin{proof}
It suffices to prove Theorem~\ref{thm:min_distortion} for a single subset of messages of size $\ell \ge k$. Fix $\delta \in (0,1/2]$, and suppose $\lambda$ satisfies
\begin{equation*}
\label{biglambda} \lambda \ge \max\left[4\left(\frac{32 \ell}{\delta p
(1-p)}\right)^{2\ell},
   \left(\frac{D_k}{\delta}\right)^2\right].
\end{equation*}
It follows from taking $Z = X$ in the definition of $\mathcal{RD}_o(Z,\gamma)$ (Definition~\ref{Me:defn}) and from the monotonicity of $\mathcal{R}_o(\mathbf{D},\lambda)$ with respect to $\lambda$ that there exist $R \in \mathbb{R}^+$ and $\gamma \in \Gamma_o$ such that, for all subsets $\mathcal{K}$ of size $k$,
\begin{equation}
\label{sumrate}
\begin{split}
D_k + \delta & \ge E[d^\lambda(X,\hat{X}_\mathcal{K})], \ \text{and} \\
kR + \delta & \ge k\mathcal{R}_o(\mathbf{D},\lambda) + \delta \ge I(X;\mathbf{U}_\mathcal{K}|T) + \sum_{i \in \mathcal{K}} I(Y_i;U_i|X,W,T).
\end{split}
\end{equation}
From (\ref{eqn:R_eq}) and (\ref{sumrate}), it follows that
\begin{equation}
\label{eqn:ob_constraint} \frac{I(X;\mathbf{U}_\mathcal{K}|T)}{k} +
\frac{1}{k}\sum_{i \in \mathcal{K}} I(Y_i;U_i|X,W,T) \le
\frac{(1-D_k)}{k} + g(D^{\frac{1}{k}}_k) + \frac{\delta}{k}.
\end{equation}
Now by the data processing inequality,
\begin{align*}
I(X;\mathbf{U}_\mathcal{K}|T) &= I(X;\mathbf{U}_\mathcal{K},T) \\
&\ge I(X;\hat{X}_\mathcal{K}).
\end{align*}

Let $\varepsilon = 1(X\cdot \hat{X}_\mathcal{K} = -1)$. We then have
\begin{align*}
I(X;\mathbf{U}_\mathcal{K}|T) & \ge H(X) - H(X|\hat{X}_\mathcal{K}) \\
       & =1 - H(X,\varepsilon|\hat{X}_\mathcal{K}) \\
       & = 1 - H(\varepsilon|\hat{X}_\mathcal{K}) - H(X|\varepsilon,\hat{X}_\mathcal{K}) \\
       & \ge 1 - h(D_k/\lambda) - \Pr(\hat{X}_\mathcal{K} = 0) \\
       & \ge (1 -D_k) - h(\delta).
\end{align*}
Using this and (\ref{eqn:ob_constraint}), we can upper bound
$\frac{1}{k}\sum_{i \in \mathcal{K}} I(Y_i;U_i|X, W, T)$ as follows:
\begin{equation}
\label{eqn:up_bound} \frac{1}{k}\sum_{i \in \mathcal{K}} I(Y_i;U_i|X, W, T)
\le g(D^{\frac{1}{k}}_k) + \frac{h(\delta)}{k} + \frac{\delta}{k}.
\end{equation}
We will show
\begin{align}
\label{eqn:mutual_inf_ub}
\frac{1}{\ell}\sum_{i=1}^\ell I(Y_i;U_i|X, W, T) \le g(D^{\frac{1}{k}}_k) + \frac{h(\delta)}{k} + \frac{\delta}{k}, \ \ell \ge k.
\end{align}
Suppose the $U_i$ are ordered according to the mutual informations $I(Y_i;U_i|X,W,T)$, i.e., we have an ordered list of messages $U_1,\ldots,U_\ell$ in which, for all $i,j \in \{1,\ldots,\ell\}, U_i$ and $U_j$ are such that $I(Y_i;U_i|X,W,T) \le I(Y_j;U_j|X,W,T)$ when $i \le j.$ The last $k$ elements of this list, $U_{\ell-k+1},\ldots,U_\ell,$ must satisfy (\ref{eqn:up_bound}), i.e.,
\begin{equation}
\label{eqn:ordered_k} \frac{1}{k}\sum_{i = \ell-k+1}^\ell
I(Y_i;U_i|Y_0, W, T) \le g(D^{\frac{1}{k}}_k) + \frac{h(\delta)}{k}
+ \frac{\delta}{k}.
\end{equation}
All other elements in the list yield equal or strictly smaller mutual informations. Therefore, if we average over a larger subset of messages, the average will never increase. We thus have
\begin{equation*}
\frac{1}{\ell}\sum_{i = 1}^\ell I(Y_i;U_i|X, W, T) \le
\frac{1}{k}\sum_{i = \ell-k+1}^\ell I(Y_i;U_i|X, W, T).
\end{equation*}
Using this and (\ref{eqn:ordered_k}), we obtain \eqref{eqn:mutual_inf_ub}. Define
\begin{align*}
\label{eqn:inv_g} \left(D_k - \zeta(D_k,\delta)\right)^{\frac{1}{k}} &=
g^{-1}\left(g(D^{\frac{1}{k}}_k)+\frac{h(\delta)}{k} +
\frac{\delta}{k}\right)
\end{align*}
for some continuous $\zeta \ge 0$ satisfying $\zeta(D_k,0)=0$. We then have
\begin{equation}
\label{eqn:bound} \frac{1}{\ell}\sum_{i=1}^\ell I(Y_i;U_i|X,W,T)
\le g((D_k-\zeta(D_k,\delta))^{\frac{1}{k}}).
\end{equation}
From (\ref{eqn:bound}), we obtain, by using Lemma~\ref{contrapositive},
\begin{align*}
D_\ell & \ge (D_k - \zeta(D_k,\delta))^{\frac{\ell}{k}} - \xi(D_\ell,\delta)
\end{align*}
for some continuous $\xi \ge 0$ satisfying $\xi(D_\ell,0) = 0$. The proof is completed by letting $\lambda \rightarrow \infty$ and then $\delta \rightarrow 0$.
\end{proof}

\appendices
\section{Preliminaries}
\label{app:prelim}
We define a multi-letter mutual information as follows:
\begin{eqnarray*}
I_K(X_1;X_2;\ldots;X_K) &=& D\left(p(X_1,\ldots,X_K)||\prod_{i=1}^K
p(X_i)\right) \\
&=& \sum_{i=1}^K H(X_i) - H(X_1,\ldots,X_K).
\end{eqnarray*}

In particular, $I_1(X) = 0$. The multi-letter mutual information, as defined above, is a measure of the mutual dependence among $K$ random variables and is different from McGill's multivariate mutual information~\cite{mcgill:multi_info_trans}. We note the following properties of $I_K(X_1;X_2;\ldots;X_K)$.

\begin{enumerate}
\item $I_K(X_1^l;\ldots;X_K^l) = \sum_{i=1}^K H(X_i^l) - H(X_1^l,\ldots,X_K^l) \geq 0.$

\item $I_K(X_1;\ldots;X_K) \geq I_m(X_1;\ldots;X_m) + I_{(K-m+1)}(f(X_1,\ldots,X_m);X_{m+1};\ldots;X_K)$, where $f(X_1,\ldots,X_m)$ is a function of the random variables $X_1,\ldots,X_m$, $m < K$.
    \newline
    Remark: This property holds by symmetry for the general case when $f(\cdot)$ is a function of any size-$m$ subset of $X_1,\ldots,X_K$.
\begin{proof}
\begin{align*}
&\phantom{=}I_K(X_1;\ldots;X_K) \\
&= \sum_{i=1}^m H(X_i) + \sum_{i=m+1}^K H(X_i) - H(X_1,\ldots,X_m) 
- H(X_{m+1},\ldots,X_K|X_1,\ldots,X_m) \\
&= I_m(X_1;\ldots;X_m) + \sum_{i=m+1}^K H(X_i) 
- H(X_{m+1},\ldots,X_K|X_1,\ldots,X_m) \\
&= I_m(X_1;\ldots;X_m) + \sum_{i=m+1}^K H(X_i) 
- H(X_{m+1},\ldots,X_K|X_1,\ldots,X_m,f(X_1,\ldots,X_m)) \\
&\ge I_m(X_1;\ldots;X_m) + \sum_{i=m+1}^K H(X_i) 
- H(X_{m+1},\ldots,X_K|f(X_1,\ldots,X_m)) \\
&= I_m(X_1;\ldots;X_m) 
+ I_{(K-m+1)}(f(X_1,\ldots,X_m);X_{m+1};\ldots;X_K),
\end{align*}
where the solitary inequality holds because conditioning never increases entropy.
\end{proof}

\item $I_K(X_1;X_2;\ldots;X_i;\ldots;X_K) \geq I_K(X_1;X_2;\ldots;f(X_i);\ldots;X_K)$, where $f(X_i)$ is a function of the random variable $X_i$. This is the data processing inequality for the multi-letter mutual information and is a special case of Property~2.
\end{enumerate}

\section{Proof of Theorem \ref{thm:conv_avg-case}}
\label{app:conv_avg-case}
The proof of the first part of Theorem~\ref{thm:conv_avg-case} is simple. Let $D_k \ge 1-\frac{k}{n}$. No excess rate for every $k$ descriptions implies that every description has rate $R$. If the decoder receives $m$ descriptions, then it receives a sum-rate of $mR$ bits per source symbol. Using the point-to-point rate-distortion function for a binary source with erasure distortion, we get $D_m \ge 1 - mR$.

The proof of the second part of Theorem~\ref{thm:conv_avg-case} is less trivial. We begin with a lemma.

\begin{definition}
Let $X$ be a binary random variable taking values in $\mathcal{X}$. An \emph{erased version} of $X$ is a random variable $\tilde{X}$, taking values in $\mathcal{\hat{X}}$, such that $\Pr(\tilde{X}=+,X=-) = \Pr(\tilde{X}=-,X=+) = 0$.
\end{definition}

\begin{lemma}
\label{lemma:cases_n_k}
Let $X_1,\ldots,X_n$ be erased versions of a uniform binary random variable $X$ taking values in $\{+,-\}$. If $\left(1 - \frac{1}{n}\right)^k \le \frac{1}{2}$ and $I_k(X_{s_1};\ldots;X_{s_k}) = 0 \quad \forall \ S = \{s_1,\ldots,s_k\},S \subset \mathcal{N}, |S| = k$, then $\sum_{i=1}^n \Pr(X_i=0) \ge n-1$.
\end{lemma}
\begin{proof}
$\left(1 - \frac{1}{n}\right)^k \le \frac{1}{2} \Rightarrow \left(\frac{1}{2}\right)^{\frac{1}{k}} \ge 1-\frac{1}{n}.$
We have the following four cases:
\newline
\textbf{Case I}: There exists $i \in \mathcal{N}$ such that $\Pr(X_i = +) > 0$ and $\Pr(X_i = -) > 0$.
\newline
Assume $i=1$ without loss of generality. Since $X_1,\ldots,X_n$ are erased versions of the same variable, they can never disagree in the source symbol they reveal (i.e., if $X_i = +$ for some $i \in \mathcal{N}$, then the rest cannot be $-$, and if $X_i = -$, then the rest cannot be $+$). Thus $\Pr(X_1 = +, X_j = -) = 0$, $j \in \{2,\ldots,n\}$. Since $I_k(X_{s_1};\ldots;X_{s_k}) = 0$ for any set of $k$ variables containing $X_1$ and $X_j$,  $X_1$ and $X_j$ must be independent. Thus
\beqa
\nonumber \Pr(X_1=+)\cdot\Pr(X_j = -) &=& \Pr(X_1 = +, X_j = -) = 0 \\
\label{eqn:pr1_ub_avg}&\Rightarrow& \Pr(X_j = -) = 0.
\eeqa
Likewise, $\Pr(X_1=-,X_j=+)=0 \Rightarrow \Pr(X_j = +) = 0$. Thus $\Pr(X_j=0)=1$ and so $\sum_{i=1}^n \Pr(X_i=0) \ge n-1$.
\newline
\textbf{Case II}: There exists $i \in \mathcal{N}$ such that $\Pr(X_i = +) > 0$ and $\Pr(X_i = -) = 0$, and Case I does not hold.
\newline
Let $S = \{s_1,\ldots,s_k\}$ be a size-$k$ subset of $\mathcal{N}$. For all $\mathcal{T} \subset S$, denote by $E_{\mathcal{T}}$ the event that $X_{s_j} = - \ \forall \ s_j \in \mathcal{T}$, and $X_{s_j}=0 \ \forall \ s_j \notin \mathcal{T}, \ s_j \in S$. Now since $\Pr(X_{s_j} = -) = 0$ from \eqref{eqn:pr1_ub_avg}, $\Pr(E_\mathcal{T}) = 0 \ \forall \ \mathcal{T}\neq \emptyset$. Thus
\beqa
\nonumber \Pr(X = -) &\le& \sum_{\mathcal{T} \subset S} \Pr(E_{\mathcal{T}}) \\
\label{eqn:source1_lb} &=& \Pr(X_{s_1}=X_{s_2}=\ldots=X_{s_k}=0).
\eeqa
Since $\Pr(X=-) = 1/2$ and $(X_{s_1},\ldots,X_{s_k})$ are independent, (\ref{eqn:source1_lb}) yields
\beqano
\prod_{j=1}^k \Pr(X_{s_j} = 0) =  \Pr(X_{s_1}=X_{s_2}=\ldots=X_{s_k}=0) \ge \frac{1}{2}.
\eeqano
In order to lower bound $\sum_{i=1}^n \Pr(X_i=0)$, we solve
\beqano
\min &\sum_{j=1}^n& \Pr(X_j=0) \\
\textrm{s.t.} \ \ &\prod_{j=1}^k& \Pr(X_{s_j} = 0) \ge \frac{1}{2} \quad \forall \ S = \{s_1,\ldots,s_k\} \subset \mathcal{N}.
\eeqano
This is a convex optimization problem, as can be readily seen by substituting $\alpha_j = \log \Pr(X_j = 0)$, and can therefore be solved by choosing $\Pr(X_j =0)=\left(\frac{1}{2}\right)^{\frac{1}{k}}$ for $j = 1,\ldots,n$. Thus $\sum_{j=1}^n \Pr(X_j=0) \ge n\left(\frac{1}{2}\right)^{\frac{1}{k}} \ge n(1-1/n) = n - 1$.
\newline
\textbf{Case III}: There exists $i \in \mathcal{N}$ such that $\Pr(X_i = -) > 0$ and $\Pr(X_i = +) = 0$, and Case I does not hold.
\newline
This case is symmetric to Case II.
\newline
\textbf{Case IV}: For all $i \in \mathcal{N}$, $\Pr(X_i = +) = \Pr(X_i = -) = 0$.
\newline
We have $\sum_{j=1}^n \Pr(X_j=0) > \sum_{j=2}^n \Pr(X_j=0) = n - 1$.
\end{proof}
We are now in a position to prove the second part of Theorem~\ref{thm:conv_avg-case}. Let $D_k < 1-\frac{k}{n}$, $D_k$ rational, and $\left(1 - \frac{1}{n}\right)^k \le \frac{1}{2}$, and let $f_i$, $i\in \mathcal{N}$ and $g_\mathcal{K}$, $\mathcal{K} \subseteq \mathcal{N}$, $\mathcal{K}\neq \emptyset$ be a code that achieves the rate-distortion vector $(R,D_1,\ldots,D_k,\ldots,D_n)$. Let $f_i, \ i \in \mathcal{N}$ have rate $R$. We have $$lR \ge H(f_i), \ i \in \mathcal{N}.$$ Let $X_\mathcal{S}^l$ be the reconstruction when the source is reconstructed from a set $\mathcal{S}$ of descriptions. Then $\forall \ S = \{s_1,\ldots,s_k\} \subset \mathcal{N},|S| = k$, we have $$H(f_{s_1}\ldots f_{s_k}) \ge H(X_\mathcal{S}^l) = l(1-D_k).$$
Thus
\begin{align*}
I_k(f_{s_1};\ldots;f_{s_k}) &= \sum_{j=1}^k H(f_{s_j}) - H(f_{s_1}\ldots f_{s_k}) \\
&\le klR - l(1-D_k) = 0.
\end{align*}

Let $X_{s_i}^l$ be the reconstruction when the decoder receives the $s_i^{th}$ description only. Then $I_k(X_{s_1}^l;\ldots;X_{s_k}^l) \le I_k(f_{s_1};\ldots;f_{s_k})= 0$ (Property~3) and so $I_k(X_{s_1,t};\ldots;X_{s_k,t}) = 0$, $t\in \{1,\ldots,l\}$. By Lemma~\ref{lemma:cases_n_k}, $\sum_{i=1}^n \Pr(X_{it}=0) \ge (n-1)$ for $t \in \{1,\ldots,l\}$. Thus
\beqano
\frac{1}{l}\sum_{t=1}^l \sum_{i=1}^n \Pr(X_{it}=0) &\ge& n-1 \\
\Rightarrow \max_{i} \left(\frac{1}{l}\sum_{t=1}^l \Pr(X_{it}=0)\right) &\ge& 1-\frac{1}{n}.
\eeqano
This completes the proof.

\section{Proof of Theorem~\ref{thm:(4,2)_avg-case}}
\label{app:(4,2)_avg-case}
We establish two lemmas before proving Theorem~\ref{thm:(4,2)_avg-case}.
\begin{lemma}
\label{lemma:case2_n_3_4}
Let $X_1,X_2,$ and $X_3$ be Bernoulli random variables such that $I(X_i;X_j) = 0$, $\forall \textrm{ } i,j \in \{1,2,3\}, i\neq j$, and $\Pr(X_1=X_2=X_3=0) \ge \frac{1}{2}$. Let $p = \max (\Pr(X_1=0),\Pr(X_2=0))$. Then
\beqano
\label{eqn:lemma_n34}
\Pr(X_3=0) \ge \frac{1}{2} + \frac{p(1-p)}{2p-1}.
\eeqano
\end{lemma}
\begin{proof}
If $p = 1$, then the conclusion follows directly from the hypothesis,
so suppose that $p < 1$.
Let $p_i$ denote $\Pr(X_i=0)$, $p(x_1,x_2,x_3)$ denote $\Pr(X_1=x_1,X_2=x_2,X_3=x_3)$, and $p_{x_3|x_1,x_2}$ denote $\Pr(X_3=x_3|X_1=x_1,X_2=x_2).$ Let $q_0=p_{0|0,0}$, $q_1=p_{0|0,1}$, and $q_2=p_{0|1,1}$. We thus have $p(0,0,0) = p_1p_2q_0$, $p(0,1,0) = p_1(1-p_2)q_1$, and $p(1,1,0) = (1-p_1)(1-p_2)q_2$. Then
\beqa
\nonumber \Pr(X_1=0,X_3=0)&=&p(0,0,0)+p(0,1,0) \\
\label{eqn:joint_bound_first} &=& p_1(p_2q_0+(1-p_2)q_1) \\
\newline
\nonumber \Pr(X_2=1,X_3=0)&=&p(0,1,0)+p(1,1,0) \\
\label{eqn:joint_bound_last} &=& (1-p_2)(p_1q_1+(1-p_1)q_2).
\eeqa
Since $(X_1,X_3)$ and $(X_2,X_3)$ are pairwise independent, we have, from \eqref{eqn:joint_bound_first}  and \eqref{eqn:joint_bound_last},
\beqa
\nonumber \Pr(X_1=0,X_3=0) &=& p_1p_3 = p_1(p_2q_0+(1-p_2)q_1) \\
\label{eqn:p3_bound_1}\Rightarrow p_3 &=& p_2q_0+(1-p_2)q_1, \\
\newline
\nonumber \Pr(X_2=1,X_3=0) &=& (1-p_2)p_3 \\
\nonumber &=& (1-p_2)(p_1q_1+(1-p_1)q_2) \\
\label{eqn:p3_bound_4}\Rightarrow p_3 &=&  p_1q_1+(1-p_1)q_2.
\eeqa
From (\ref{eqn:p3_bound_1}) and (\ref{eqn:p3_bound_4}),
\beqa
\nonumber p_1q_1+(1-p_1)q_2 &=& p_2q_0+(1-p_2)q_1 \\
\Rightarrow\label{eqn:q2}q_2 &=& \frac{p_2q_0 - (p_1+p_2-1)q_1}{1-p_1}.
\eeqa
Since $p(0,0,0) \ge 1/2$ by hypothesis, we have $p_1 p_2 \ge 1/2$,
and thus $p_1+p_2-1>0$. Now since $q_2 \le 1$, \eqref{eqn:q2} gives
\beqa
1 \ge \frac{p_2q_0 - (p_1+p_2-1)q_1}{1-p_1} \label{eqn:q1_bound}\Rightarrow q_1 \ge \frac{p_2q_0 - (1-p_1)}{p_1+p_2-1}.
\eeqa
Now
\beqa
p(0,0,0) = p_1p_2q_0  \ge \frac{1}{2} \label{eqn:p2q1_bound}\Rightarrow p_2q_0 \ge \frac{1}{2p_1}.
\eeqa
Assume without loss of generality that $p_1 \ge p_2$. Then $p_1 + p_2 \le 2p_1$. Substituting this and (\ref{eqn:p2q1_bound}) into \eqref{eqn:q1_bound} yields
\beqa
\label{eqn:q2_final_bound}
q_1 \ge \frac{\frac{1}{2p_1} - 1+p_1}{2p_1-1} = \frac{p_1}{2p_1-1} - \frac{1}{2p_1}.
\eeqa
Upon substituting (\ref{eqn:p2q1_bound}) and \eqref{eqn:q2_final_bound} into (\ref{eqn:p3_bound_1}), we get
\beqano
p_3 &\ge& \frac{1}{2p_1} + (1-p_2)\left(\frac{p_1}{2p_1-1} - \frac{1}{2p_1}\right) \\
&\ge& \frac{1}{2p_1} + (1-p_1)\left(\frac{p_1}{2p_1-1} - \frac{1}{2p_1}\right) \\
&=& \frac{1}{2} + \frac{p_1(1-p_1)}{2p_1-1}
\eeqano
where the last inequality follows because $p_2 \le p_1$ and $\frac{p_1}{2p_1-1} - \frac{1}{2p_1} > 0$.
\end{proof}

\begin{corollary}
\label{cor:max4}
Let $X_1,X_2,X_3$ and $X_4$ be Bernoulli random variables such that $I(X_i;X_j) = 0$, $\forall \textrm{ } i,j \in \{1,2,3,4\}, i\neq j$, and $\Pr(X_1=X_2=X_3=X_4=0) \ge \frac{1}{2}$. Then $$\sum_{i=1}^4 \Pr(X_i=0) \ge 3.$$
\end{corollary}
\begin{proof}
Let $p_i = \Pr(X_i=0)$. Assume WLOG that $p_1 \ge p_2 \ge p_3 \ge p_4$. Now $p_3p_4 = \Pr(X_3=X_4=0) \ge 1/2$ by hypothesis, which implies $p_3 \ge 1/\sqrt{2}$ and $p_4 \ge 1/2p_3$. Applying Lemma~\ref{lemma:case2_n_3_4} to $X_2$, $X_3$, and $X_4$ gives $p_2 \ge \frac{1}{2} + \frac{p_3(1-p_3)}{2p_3-1}$. Thus
\begin{align*}
\sum_{i=1}^4 p_i &= p_1 + p_2 + p_3 + p_4  \\
&\ge 2p_2 + p_3 + p_4 \\
&\ge 2 \max \left(p_3,\frac{1}{2} + \frac{p_3(1-p_3)}{2p_3-1}\right) + p_3 + \frac{1}{2p_3} \\
&\ge \min_{x \in [\frac{1}{\sqrt{2}},1]} 2 \max \left(x,\frac{1}{2} + \frac{x(1-x)}{2x-1}\right) + x + \frac{1}{2x}.
\end{align*}
Since $\frac{1}{2} + \frac{p_3(1-p_3)}{2p_3-1}$ is monotonically decreasing in $p_3$ for $p_3 \in (1/2,1]$, it is easy to verify that
\begin{equation*}
\max \left(x,\frac{1}{2} + \frac{x(1-x)}{2x-1}\right) = \left\{
\begin{array}{ll}
x & \text{if } x \ge \frac{1}{2} + \frac{1}{\sqrt{12}}\\
\frac{1}{2} + \frac{x(1-x)}{2x-1} & \text{if } x \le \frac{1}{2} + \frac{1}{\sqrt{12}},
\end{array} \right.
\end{equation*}
where $\frac{1}{2} + \frac{1}{\sqrt{12}}$ is the admissible solution to the equation $x = \frac{1}{2} + \frac{x(1-x)}{2x-1}$. Thus
\begin{align*}
\sum_{i=1}^4 p_i &\ge \min \left(\min_{x \in [\frac{1}{\sqrt{2}},\frac{1}{2}+\frac{1}{\sqrt{12}}]} 2\left(\frac{1}{2} + \frac{x(1-x)}{2x-1}\right) + x + \frac{1}{2x}, \min_{x \in [\frac{1}{2}+\frac{1}{\sqrt{12}},1]} 2x + x + \frac{1}{2x}\right) \\
&= \min \left(\min_{x \in [\frac{1}{\sqrt{2}},\frac{1}{2}+\frac{1}{\sqrt{12}}]} 1 + \frac{1}{2x} + \frac{x}{2x-1}, \min_{x \in [\frac{1}{2}+\frac{1}{\sqrt{12}},1]} 3x + \frac{1}{2x}\right) \\
&= \min (3,3) = 3,
\end{align*}
where the penultimate inequality follows from the fact that $1 + \frac{1}{2x} + \frac{x}{2x-1}$ is a monotonically decreasing in $x$ for $x \in [\frac{1}{\sqrt{2}},\frac{1}{2}+\frac{1}{\sqrt{12}}]$ and takes a minimum value of 3 at $x = \frac{1}{2}+\frac{1}{\sqrt{12}}$, and that $3x + \frac{1}{2x}$ is monotonically increasing in $x$ for $x \in [\frac{1}{2}+\frac{1}{\sqrt{12}},1]$ and takes a minimum value of 3 at $x = \frac{1}{2}+\frac{1}{\sqrt{12}}$.
\end{proof}

\begin{lemma}
\label{lemma:cases_n=4_k=2}
Let $X_1,\ldots,X_4$ be erased versions of a uniform binary random variable $X$ taking values in $\{+,-\}$. If $I(X_i;X_j) = 0, \ i,j \in \{1,\ldots,4\}, \ i \neq j$, then
\begin{align*}
\sum_{i = 1}^4 \Pr(X_i=0) \ge 3.
\end{align*}
\end{lemma}
\begin{proof}
We have the four cases as in the proof of Lemma~\ref{lemma:cases_n_k}:
\newline
\textbf{Case I}: There exists $i \in \{1,2,3,4\}$ such that $\Pr(X_i = +) > 0$ and $\Pr(X_i = -) > 0$.
\newline
Just as in the proof of Lemma~\ref{lemma:cases_n_k}, we have $\sum_{j=1}^4 \Pr(X_j=0) \ge 4-1 = 3$.
\newline
\textbf{Case II}: There exists $i \in \{1,2,3,4\}$ such that $\Pr(X_i = +) > 0$ and $\Pr(X_i = -) = 0$, and Case I does not hold.
\newline
Assume $i=1$ WLOG. Then from \eqref{eqn:pr1_ub_avg}, $\Pr(X_j = -) = 0$ for $j \in \{2,3,4\}$. Thus the $X_j$ are effectively binary random variables such that $\Pr(X_1=\ldots=X_4=0)\ge 1/2$. By Corollary~\ref{cor:max4}, $\sum_{j = 1}^4 \Pr(X_j=0) \ge 3$.
\newline
\textbf{Case III}: There exists $i \in \{1,2,3,4\}$ such that $\Pr(X_i = -) > 0$ and $\Pr(X_i = +) = 0$, and Case I does not hold.
\newline
This case is analogous to Case II.
\newline
\textbf{Case IV}: For all $i \in \{1,2,3,4\}$, $\Pr(X_i = +) = \Pr(X_i = -) = 0$.
\newline
We have $\sum_{j=1}^4 \Pr(X_j=0) > \sum_{j=2}^4 \Pr(X_j=0) = 4 - 1 = 3$.
\end{proof}

We are now in a position to prove Theorem~\ref{thm:(4,2)_avg-case}. Let $f_i$, $i\in \mathcal{N}$ and $g_\mathcal{K}$, $\mathcal{K} \subseteq \mathcal{N}$ be a code that achieves $(\frac{1-D_2}{2},D_1,D_2,D_3,D_4)$. Using the same argument as that in the proof of the second part of Theorem~\ref{thm:conv_avg-case}, we have for $i,j \in \{1,2,3,4\}$, $i \neq j$ that $I(X_i^l;X_j^l) \le I(f_i;f_j) = 0$ and thus $I(X_{it};X_{jt}) = 0$ for all $t \in \{1,\ldots,l\}$. By Lemma~\ref{lemma:cases_n_k=2}, $\sum_{i=1}^4 \Pr(X_{it}=0) \ge 3$ for $t \in \{1,\ldots,l\}$. It follows that
\beqano
\frac{1}{l}\sum_{t=1}^l \sum_{i=1}^4 \Pr(X_{it}=0) &\ge& 3 \\
\Rightarrow \max_{i} \left(\frac{1}{l}\sum_{t=1}^l \Pr(X_{it}=0)\right) &\ge& \frac{3}{4}.
\eeqano
This completes the proof.

\section{Proof of Theorem \ref{thm:chebyshev_avg-case}}
\label{app:chebyshev_avg-case}
We establish two lemmas before proving Theorem~\ref{thm:chebyshev_avg-case}.
\begin{lemma}
\label{lemma:chebyshev_bound}
Let $X_1,\ldots,X_n$ be Bernoulli random variables such that $I(X_i;X_j) = 0$ $\forall \textrm{ } i,j \in \mathcal{N}, \ i\neq j$, and $\Pr(X_1=X_2=\ldots =X_n=0) \ge \frac{1}{2}$. Then
\beqano
\frac{1}{n}\sum_{i=1}^n \Pr(X_i=0) \ge 1 - \frac{2}{n}.
\eeqano
\end{lemma}
\begin{proof}
Let $p_i$ denote $\Pr(X_i=0)$ and let $q_i=\Pr(X_i=1)=1-p_i$. Since the $X_i$'s are pairwise independent, we have
\begin{align*}
\bE \left[\frac{1}{n}\sum_{i=1}^n X_i\right] &= \frac{1}{n}\sum_{i=1}^n q_i \\
\textrm{Var} \left[\frac{1}{n}\sum_{i=1}^n X_i\right] &=  \frac{1}{n^2}\sum_{i=1}^n \textrm{Var}(X_i) = \frac{1}{n^2}\sum_{i=1}^n p_iq_i.
\end{align*}
Let $\alpha > \sqrt{\frac{2}{n^2}(\sum_{i=1}^n p_iq_i)}.$ Then, by Chebyshev's inequality,
\beqano
\Pr \left(\left|\frac{1}{n}\sum_{i=1}^n X_i - \frac{1}{n}\sum_{i=1}^n q_i\right| > \alpha \right) &\le& \frac{\textrm{Var}\left[\frac{1}{n}\sum_{i=1}^n X_i\right]}{\alpha^2} \\
&=& \frac{\sum_{i=1}^n p_iq_i}{n^2 \alpha^2} < \frac{1}{2}.
\eeqano
Let $E_1$ and $E_2$ be the events $|\frac{1}{n}\sum_{i=1}^n X_i - \frac{1}{n}\sum_{i=1}^n q_i| \le \alpha$ and $X_1=X_2=\ldots =X_n=0$, respectively. Then $\Pr(E_1) > \frac{1}{2}$, and $\Pr(E_2) \ge \frac{1}{2}$ by hypothesis. Since $\Pr(E_1)+\Pr(E_2)>1$, $\Pr(E_1 \cap E_2) > 0.$ This implies that
\beqano
\frac{1}{n}\sum_{i=1}^n q_i \le \alpha \Rightarrow \frac{1}{n}\sum_{i=1}^n p_i &\ge& 1- \alpha.
\eeqano
Since $\alpha$ was arbitrary, this implies
\beqa
   \frac{1}{n}\sum_{i=1}^n p_i &\ge& 1 - \sqrt{\frac{2}{n^2}(\sum_{i=1}^n p_iq_i)}.
\label{eqn:max_pi_n}
\eeqa
Moreover,
\beqano
\frac{1}{n}\sum_{i=1}^n p_iq_i &\le& \frac{1}{n}\sum_{i=1}^n q_i \le \sqrt{\frac{2}{n^2}(\sum_{i=1}^n p_iq_i)}.
\eeqano
A little algebra gives
\beqa
\label{eqn:pq_bound}
\sum_{i=1}^n p_iq_i \le \sqrt{2\sum_{i=1}^n p_iq_i} \Rightarrow \sum_{i=1}^n p_iq_i \le 2.
\eeqa
Substituting (\ref{eqn:pq_bound}) into (\ref{eqn:max_pi_n}) yields
\beqano
\frac{1}{n}\sum_{i=1}^n p_i &\ge& 1 -\sqrt{\frac{2}{n^2}\cdot 2} = 1 - \frac{2}{n}.
\eeqano
\end{proof}

\begin{lemma}
\label{lemma:cases_n_k=2}
Let $X_1,\ldots,X_n$ be erased versions of a uniform binary random variable $X$ taking values in $\{+,-\}$. If $I(X_i;X_j) = 0, \ i,j \in \mathcal{N}, \ i \neq j$, then
\begin{align*}
\sum_{i=1}^n \Pr(X_i=0) \ge n - 2.
\end{align*}
\end{lemma}
\begin{proof}
We have Cases I, II, III, and IV as in the proof of Lemma~\ref{lemma:cases_n_k}. Cases I and IV are the same as those in Lemma~\ref{lemma:cases_n_k}, so we will just mention Cases II and III.
\newline
\textbf{Case II}: There exists $i \in \mathcal{N}$ such that $\Pr(X_i = +) > 0$ and $\Pr(X_i = -) = 0$ and Case I does not hold.
\newline
Assume $i=1$ WLOG. Then from \eqref{eqn:pr1_ub_avg}, $\Pr(X_j = -) = 0$ for $j \in \{2,\ldots,n\}$. Thus the $X_j$'s are always erased when the binary source $X=-$, and so $\Pr(X_1=\ldots=X_n=0)\ge 1/2$. By Lemma~\ref{lemma:chebyshev_bound}, $\sum_{i=1}^n \Pr(X_i=0) \ge n-2$. The proof of Case III is analogous to the proof of Case II.
\end{proof}

We are now in a position to prove Theorem~\ref{thm:chebyshev_avg-case}. Let $f_i$, $i\in \mathcal{N}$ and $g_\mathcal{K}$, $\mathcal{K} \subseteq \mathcal{N}$ be a code that achieves $(\frac{1-D_2}{2},D_1,D_2,\ldots,D_n)$. Using the same argument as that in the proof of the second part of Theorem~\ref{thm:conv_avg-case}, we have for $i,j \in \mathcal{N}$, $i \neq j$ that $I(X_i^l;X_j^l) \le I(f_i;f_j) = 0$ and thus $I(X_{it};X_{jt}) = 0$ for $t \in \{1,\ldots,l\}$. By Lemma~\ref{lemma:cases_n_k=2}, $\sum_{i=1}^n \Pr(X_{it}=0) \ge n - 2$ for $t \in \{1,\ldots,l\}$. It follows that
\beqano
\frac{1}{l}\sum_{t=1}^l \sum_{i=1}^n \Pr(X_{it}=0) &\ge& n - 2. \\
\Rightarrow \max_{i} \left(\frac{1}{l}\sum_{t=1}^l \Pr(X_{it}=0)\right) &\ge& 1 - \frac{2}{n}.
\eeqano
This completes the proof.

\section{Proof of Lemma~\ref{lemma:cases}}
\label{app:lemma_cases}
For any $t \in \{1,\ldots,l\}$, we have exactly one of the following four cases:
\newline
\textbf{Case I}: $\exists \ i \in \mathcal{N}$ s.t. $\Pr(X_{it}(X) = +) > 0$ and $\Pr(X_{it}(X) = -) > 0$.
\newline
\textbf{Case II}: $\exists \ i \in \mathcal{N}$ s.t. $\Pr(X_{it}(X) = +) > 0$ and $\Pr(X_{it}(X) = -) = 0$, and Case I does not hold.
\newline
\textbf{Case III}: $\exists \ i \in \mathcal{N}$ s.t. $\Pr(X_{it}(X) = -) > 0$ and $\Pr(X_{it}(X) = +) = 0$, and Case I does not hold.
\newline
\textbf{Case IV}: $\forall \ i \in \mathcal{N}, \ \Pr(X_{it}(X) = +) = \Pr(X_{it}(X) = -) = 0$.

Let $\mathcal{B}_1$, $\mathcal{B}_2$, $\mathcal{B}_3$ and $\mathcal{B}_4$ be the sets of $t \in \{1,\ldots,l\}$ satisfying Cases I, II, III and IV, respectively. Moreover, let $|\mathcal{B}_1| = b_1$, $|\mathcal{B}_2| = b_2$, $|\mathcal{B}_3| = b_3$ and $|\mathcal{B}_4| = b_4$. Then $b_1 + b_2 + b_3 + b_4 = l$. Now consider a source string $(x^*)^l$ such that $x^*_t = -$ if $t \in \mathcal{B}_2$ and $x^*_t = +$ if $t \in \mathcal{B}_3$. We have
\begin{align*}
&\phantom{=}\max_{x^l \in \mathcal{X}^l} \sum_{i=1}^n \left[\frac{1}{l}\sum_{t=1}^l d(x_t,X_{it}(x))\right] \\
&\ge \sum_{i=1}^n \frac{1}{l}\sum_{t=1}^l d(x^*_t,X_{it}(x^*)) \\
&= \frac{1}{l}\sum_{t \in \mathcal{B}_1} \sum_{i=1}^n d(x^*_t,X_{it}(x^*)) + \frac{1}{l}\sum_{t \in \mathcal{B}_2} \sum_{i=1}^n d(x^*_t,X_{it}(x^*)) \\ &\phantom{+} + \frac{1}{l}\sum_{t \in \mathcal{B}_3} \sum_{i=1}^n d(x^*_t,X_{it}(x^*)) + \frac{1}{l}\sum_{t \in \mathcal{B}_4} \sum_{i=1}^n d(x^*_t,X_{it}(x^*)).
\end{align*}
Consider now $t \in \mathcal{B}_1$. Since $X_{1t}(X),\ldots,X_{nt}(X)$ are erased versions of the same binary random variable $X_t$, they can never disagree in the source symbol they reveal. We therefore have $\Pr(X_{it}(X) = +, X_{jt}(X) = -) = 0$, $j \in \mathcal{N}, \ j \neq i$. Since $X_{it}(X)$ and $X_{jt}(X)$, $i,j \in \mathcal{N}$, $i \neq j$, are pairwise independent, we have $\Pr(X_{it}(X) = +)\cdot \Pr(X_{jt}(X) = -)$
\beqa
\nonumber &=& \Pr(X_{it}(X) = +, X_{jt}(X) = -) = 0 \\
\label{eqn:pr1_ub}&\Rightarrow& \Pr(X_{jt}(X) = -) = 0,
\eeqa
since $\Pr(X_{it}(X) = +) > 0$. Repeating the same analysis with $\Pr(X_{it}(X) = -, X_{jt}(X) = +)$ yields $\Pr(X_{jt}(X) = +) = 0$. Thus $\Pr(X_{jt}(X) = 0) = 1$ for all $j \in \mathcal{N}, \ j \neq i$, and therefore $X_{jt}(x^*) = 0$ for all $j \in \mathcal{N}, \ j \neq i$. Similarly, it follows from \eqref{eqn:pr1_ub} that $\Pr(X_{jt}(X) = -) = 0$ for $j \in \mathcal{N}, \ j \neq i$ if $t \in \mathcal{B}_2$ and $\Pr(X_{jt}(X) = +) = 0$ for $j \in \mathcal{N}, \ j \neq i$ if $t \in \mathcal{B}_3$. Thus by construction, $X_i^l(x^*)$, $i \in \mathcal{N}$, must have $X_{it}(x^*) = 0$ for $t \in \mathcal{B}_2 \cup \mathcal{B}_3 \cup \mathcal{B}_4$. It follows that
\begin{align*}
&\phantom{=}\max_{x^l \in \mathcal{X}^l} \sum_{i=1}^n \left[\frac{1}{l}\sum_{t=1}^l d(x_t,X_{it}(x))\right] \\
&\ge \frac{1}{l}\sum_{t \in \mathcal{B}_1} \sum_{i=1}^n 1_{(X_{it}(x^*)=0)} + \frac{1}{l}\sum_{t \in \mathcal{B}_2} \sum_{i=1}^n 1_{(X_{it}(x^*)=0)} \\ &\phantom{+} \qquad + \frac{1}{l}\sum_{t \in \mathcal{B}_3} \sum_{i=1}^n 1_{(X_{it}(x^*)=0)} + \frac{1}{l}\sum_{t \in \mathcal{B}_4} \sum_{i=1}^n 1_{(X_{it}(x^*)=0)} \\
&\ge \frac{1}{l}b_1(n-1) + \frac{1}{l}b_2n + \frac{1}{l}b_3n + \frac{1}{l}b_4n \\
&= \frac{1}{l}(nl-b_1) \\
&= n - \frac{b_1}{l} \ge n - 1.
\end{align*}

This completes the proof.

\section{Proof of Theorem~\ref{thm:conv_w-case_D_1}}
\label{app:conv_w-case_D_1}
Let $D_k < 1-\frac{k}{n}$ and rational. Let $f_i$, $i\in \mathcal{N}$ and $g_\mathcal{K}$, $\mathcal{K} \subseteq \mathcal{N}, \mathcal{K} \neq \emptyset$, be a code that achieves $(R,D_1,\ldots,D_k,\ldots,D_n)$. Let $R$ be the rate of $f_i, \ i \in \mathcal{N}$. Consider endowing the source with an \emph{i.i.d.} uniform distribution over $\mathcal{X}^l$ for analysis purposes. Then for all $i \in \mathcal{N}$,
\begin{align}
\label{eqn:H(f_i)_bound}
lR \ge H(f_i).
\end{align}
Let $\hat{X}_\mathcal{S}^l$ be the reconstruction when the source is reconstructed from a set $\mathcal{S}$ of descriptions. Then $\forall \ S = \{s_1,\ldots,s_k\} \subset \mathcal{N},|S| = k$, we have $$H(f_{s_1}\ldots f_{s_k}) \ge H(\hat{X}_\mathcal{S}^l) \ge I(X^l;\hat{X}_\mathcal{S}^l)\ge l(1-D_k),$$ where the final inequality follows because the average distortion is no lower than the worst-case distortion.
Thus
\begin{align}
\nonumber I_k(f_{s_1};\ldots;f_{s_k}) &= \sum_{j=1}^k H(f_{s_j}) - H(f_{s_1}\ldots f_{s_k}) \\
\label{eqn:indep_messages}&\le klR - l(1-D_k) = 0.
\end{align}

Let $\hat{X}_{s_i}^l$ be the reconstructed source string when the decoder has access to the $s_i^{th}$ description only. By Property~3 of the multi-letter mutual information, $ I_k(\hat{X}_{s_1}^l;\ldots;\hat{X}_{s_k}^l) \le I_k(f_{s_1};\ldots;f_{s_k})= 0$ for all $S \subset \mathcal{N}, \ |S| = k$. By Property~2 of the multi-letter mutual information, $I(\hat{X}_i^l;\hat{X}_j^l) = 0$ for all $i,j \in \mathcal{N}, \ i \neq j$, and thus $I(\hat{X}_{it};\hat{X}_{jt}) = 0$ for all $i,j \in \mathcal{N}, \ i \neq j$, and $t = 1,\ldots,l$. Now if any two of the $\hat{X}_{s_i}^l$ disagree in a source symbol they reveal, then the resulting single-message distortion is going to be $\infty$ and the result follows trivially, so suppose that the $\hat{X}_{s_i}^l$ are consistent. Then by Lemma~\ref{lemma:cases}, we have
\begin{align*}
\sum_{i=1}^n \max_{x^l \in \mathcal{X}^l}\left[\frac{1}{l}\sum_{t=1}^l d(x_t,\hat{X}_{it})\right] \ge n-1,
\end{align*}
which implies
\begin{align*}
D_1 = \max_{i \in \mathcal{N}} \max_{x^l \in \mathcal{X}^l} \left[\frac{1}{l}\sum_{t=1}^l d(x_t,\hat{X}_{it})\right] \ge \frac{n-1}{n} = 1-\frac{1}{n}.
\end{align*}

This completes the proof.

\section{Proof of Theorem~\ref{thm:pareto_optimality}}
\label{app:pareto_optimality}
Consider $R$ first. If $R < \frac{1-D_k}{k}$, then the sum rate of any $k$ descriptions is strictly less than $1-D_k$, and the source string cannot be reconstructed with distortion $D_k$. Thus the rate of each description must be at least $\frac{1-D_k}{k}$. Now, in light of the previous theorem, it suffices to show that for any $(R,D_1,\ldots,D_k,\ldots,D_n) \in \mathcal{RD}_{worst}$, if $D_1 = 1-\frac{1}{n}$, then $D_m \ge 1-\frac{m}{n}$ for $m < k$. Let $S = \{s_1,\ldots,s_k\}$ and $\mathcal{M} = \{s_1,\ldots,s_m\}$. Let $X_\mathcal{M}^l$ be the source reconstruction when the decoder has access to set of descriptions indexed by the elements in $\mathcal{M}$. Then from \eqref{eqn:indep_messages} and Properties 2 and 3 of the multi-letter mutual information, it follows that
\begin{align*}
I(X_\mathcal{M}^l;X_{s_{m+1}}^l;\ldots;X_{s_k}^l) &\le I(X_\mathcal{M}^l; f_{s_{m+1}};\ldots;f_{s_k}) \\
&\le I_k(f_{s_1};\ldots;f_{s_k}) = 0,
\end{align*}
and thus $I(X_{\mathcal{M},t};X_{s_{m+1},t};\ldots,X_{s_k,t}) = 0$ for $t = 1,\ldots,l$. This implies that for each $t$, the $(n-m+1)$ random variables $\{X_{\mathcal{M},t};X_{s_{m+1},t};\ldots;X_{s_n,t}\}$ are pairwise independent, and therefore by Lemma~\ref{lemma:cases},
\begin{align*}
\max_{x^l \in \mathcal{X}^l}\left[\frac{1}{l}\sum_{t=1}^l d(x_t,X_{\mathcal{M},t})\right] + \sum_{i=m+1}^n \max_{x^l \in \mathcal{X}^l}\left[\frac{1}{l}\sum_{t=1}^l d(x_t,X_{s_i,t})\right] \ge n-m.
\end{align*}
Since $D_1 = 1-\frac{1}{n}$, we have
$$\max_{x^l \in \mathcal{X}^l} \left[\frac{1}{l}\sum_{t=1}^l d(x_t,X_{s_i,t})\right] \le 1-\frac{1}{n}$$ for $m+1 \le i \le n$, and thus
\begin{align*}
\max_{x^l \in \mathcal{X}^l}\left[\frac{1}{l}\sum_{t=1}^l d(x_t,X_{\mathcal{M},t})\right] &\ge n-m - \sum_{i=m+1}^n \max_{x^l \in \mathcal{X}^l}\left[\frac{1}{l}\sum_{t=1}^l d(x_t,X_{s_i,t})\right] \\
&\ge n-m - (n-m)\left(1-\frac{1}{n}\right) \\
&= \frac{n-m}{n} = 1-\frac{m}{n},
\end{align*}
which implies
\begin{align*}
D_m = \max_{\substack{\mathcal{M} \subset \mathcal{N} \\ |\mathcal{M}|=m}} \max_{x^l \in \mathcal{X}^l}\left[\frac{1}{l}\sum_{t=1}^l d(x_t,X_{\mathcal{M},t})\right] \ge 1-\frac{m}{n}.
\end{align*}

This completes the proof.

\section{Proof of Theorem~\ref{thm:other_wc_distortions}}
\label{app:other_wc_distortions}
Since $m$ divides $n$, we can form $n/m$ sets consisting of $m$ messages each. Denote these sets by $\mathcal{M}_1,\ldots,\mathcal{M}_{n/m}$, where $\mathcal{M}_i \subset \{f_1,\ldots,f_n\}$, $|\mathcal{M}_i|=m$, and $\mathcal{M}_i \cap \mathcal{M}_j = \emptyset$, $i,j \in \{1,\ldots,n/m\}$, $i \neq j$. Since $m \le k/2$, there exists a set $S = \{s_1,\ldots,s_k\}$ of $k$ messages containing $\mathcal{M}_i$ and $\mathcal{M}_j$ for some $i,j \in \{1,\ldots,n/m\}$, $i \neq j$. Let $X_{\mathcal{M}_i}^l$ be the source reconstruction when the decoder has access to the messages in $\mathcal{M}_i$ only. By Property~2 of the multi-letter mutual information, it follows that for the set $S$ containing $\mathcal{M}_i$ and $\mathcal{M}_j$,
\begin{align*}
I(X_{\mathcal{M}_i}^l;X_{\mathcal{M}_j}^l) &\le I_{(k-2m+2)}(X_{\mathcal{M}_i}^l;X_{\mathcal{M}_j}^l;f_r;\ldots;f_{r+k-2m-1}) \\
&\le I_k(f_{s_1};\ldots;f_{s_k}) = 0,
\end{align*}
where $f_r,\ldots,f_{r+k-2m-1} \in \{f_{s_1},\ldots,f_{s_k}\}\setminus \{\mathcal{M}_i,\mathcal{M}_j\}$. By Lemma~\ref{lemma:cases}, we have
\begin{align*}
\sum_{i=1}^{n/m} \max_{x^l \in \mathcal{X}^l}\left[\frac{1}{l}\sum_{t=1}^l d(x_t,X_{\mathcal{M}_i,t})\right] \ge \frac{n}{m}-1,
\end{align*}
and thus
\begin{align*}
D_m &= \max_{\substack{\mathcal{M} \subset \mathcal{N} \\ |\mathcal{M}|=m}} \max_{x^l \in \mathcal{X}^l}\left[\frac{1}{l}\sum_{t=1}^l d(x_t,X_{\mathcal{M},t})\right] \\
&\ge \max_{i \in \{1,\ldots,n/m\}} \max_{x^l \in \mathcal{X}^l}\left[\frac{1}{l}\sum_{t=1}^l d(x_t,X_{\mathcal{M}_i,t})\right] \\
&\ge \frac{\frac{n}{m}-1}{\frac{n}{m}} = 1 - \frac{m}{n}.
\end{align*}

This completes the proof.

\section{Proof of Theorem~\ref{main}}
\label{app:outer_bound}
This bound differs only slightly from the outer bound proposed in \cite{wagnervenkat05} and much of the proof is similar to that in \cite{wagnervenkat05}. Suppose $(\mathbf{R},\mathbf{D})$ is achievable. Let $f_1^{(l)},\ldots,f_n^{(l)}$ be encoders and $(g_\mathcal{K}^j)^l, \ \mathcal{K} \subseteq \mathcal{N}$ be decoders satisfying~(\ref{constraints}). Take any $Z$ in $\psi$ and augment the sample space to include $Z^l$ so that $(Z_t,Y_{0,t},\mathbf{Y}_{\mathcal{N},t},Y_{n+1,t})$ is independent over $t \in \{1,\ldots,l\}$. Next let $T$ be uniformly distributed over $\{1,\ldots,l\}$ and independent of $Z^l$, $Y_0^l$, $\mathbf{Y}_\mathcal{N}^l$ and $Y_{n+1}^l$. Then define
\begin{align*}
Z & = Z_T \\
Y_0 & = Y_{0,T} \\
Y_i & = Y_{i,T} \ \text{for $i \in \mathcal{N}$} \\
Y_{n+1} &= Y_{n+1,T} \\
U_i & = \left(f_i^{(l)}(Y_i^l), Z_{1:T-1},\{Y_{n+1}^l\}\backslash \{Y_{n+1,T}\}\right) \ \text{for $i \in \mathcal{N}$} \\
V_j & = V_{j,T}  \ \text{for $j = 1,\ldots,J$}\\
W & = (\{Z^l\}\backslash \{Z_T\},\{Y_{n+1}^l\}\backslash \{Y_{n+1,T}\}).
\end{align*}
It can be verified that $\gamma = (\mathbf{U}_\mathcal{N},V_1,\ldots,V_j,W,T)$ is in $\Gamma_o$ and that, together with $Y_0$, $\mathbf{Y}_\mathcal{N}$, $Y_{n+1}$, and $Z$, it satisfies the Markov coupling. It suffices to show that $(\mathbf{R},\mathbf{D})$ is in $\mathcal{RD}_o(Z,\gamma)$. Note that~(\ref{constraints}) implies
\begin{align*}
D_{k,j} &\ge \max_{\mathcal{K}:|\mathcal{K}|=k} \mathbf{E}[d_j(Y_{0,T},\mathbf{Y}_{\mathcal{K},T},Y_{n+1,T},V_{j,T})] \ \text{for } j = 1,\ldots,J,
\end{align*}
i.e.,
\begin{align*}
D_{k,j} &\ge \max_{\mathcal{K}:|\mathcal{K}|=k} \mathbf{E}[d_j(Y_0,\mathbf{Y}_\mathcal{K},Y_{n+1},V_j)] \ \text{for } j = 1,\ldots,J.
\end{align*}
Second, by the cardinality bound on entropy and the fact that conditioning never increases entropy,
\begin{align}
\nonumber l \sum_{i \in \mathcal{K}} R_i & \ge
        H\left(\left(f_i^{(l)}(Y_i^l)\right)_{i \in \mathcal{K}}\right) \\
\label{blockchain}
    & = I\left(Z^l, \mathbf{Y}_\mathcal{K}^l;
      \left(f_i(Y_i^l)\right)_{i \in \mathcal{K}}\Big|Y_{n+1}^l\right).
\end{align}
By the chain rule for mutual information,
\begin{align*}
\label{twoterms}
 I\left(Z^l, \mathbf{Y}_\mathcal{K}^l ;
    \left(f_i(Y_i^l)\right)_{i \in \mathcal{K}}\Big|Y_{n+1}^l\right) = \ &I\left(Z^l; \left(f_i(Y_i^l)\right)_{i \in \mathcal{K}}\Big|Y_{n+1}^l\right)
 + I\left(\mathbf{Y}_\mathcal{K}^l; \left(f_i(Y_i^l)\right)_{i \in \mathcal{K}} \Big|Z^l,Y_{n+1}^l\right).
\end{align*}
The rest of the proof is similar to that in \cite{wagnervenkat05}. The main difference between this proof and the proof in \cite{wagnervenkat05} is that here we do not condition on $\left(f_i(Y_i^l)\right)_{i \in \mathcal{K}^c}$ in \eqref{blockchain}. Taking the maximum over this bound and the bound in \cite{wagnervenkat05} yields the desired outer bound.

\section{Proof of Lemma~\ref{contrapositive}}
\label{app:contrapositive}
Assume WLOG that $\mathcal{K} = \{1,\ldots,\ell\}$. For each possible realization $(w,t)$ of $(W,T)$, let
\begin{equation*}
D_{w,t} = E[d^\lambda(X,\hat{X}_\mathcal{K})|W = w, T = t].
\end{equation*}
Let $S = \{(w,t) : D_{w,t} \le \sqrt{\lambda}\}$. Then by Markov's inequality,
\begin{equation}
\label{markov} \Pr((W,T) \notin S) \le
\frac{\tilde{D}}{\sqrt{\lambda}} \le \delta.
\end{equation}
In particular, $\Pr((W,T) \in S) > 0$.
Also, for any $(w,t) \in S$,
\begin{equation*}
\frac{32 \ell}{p(1-p)} \left(\frac{2
D_{w,t}}{\lambda}\right)^{1/\ell} \le \delta.
\end{equation*}
Thus, by Lemma~\ref{continuity}, if $(w,t) \in S$,
\begin{equation*}
\frac{1}{\ell} \sum_{i = 1}^\ell I(Y_i;U_i|X, W = w, T = t) \ge
 g\left((D_{w,t} + \delta)^{1/\ell}\right) + 2 \delta \log \frac{\delta}{5}.
\end{equation*}
By averaging over $(w,t) \in S$ and invoking Corollary~\ref{convexcor}, we obtain
\begin{align*}
\sum_{(w,t) \in S}
\frac{1}{\ell} \sum_{i = 1}^\ell I(Y_i;U_i|X, W = w,& T = t) \cdot
\frac{\Pr(W = w, T = t)}{\Pr((W,T) \in S)} \\
& \ge g((\tilde{D} + \delta)^{1/\ell}) + 2 \delta \log
\frac{\delta}{5}.
\end{align*}
Therefore,
\begin{align*}
\frac{1}{\ell} \sum_{i = 1}^\ell I(Y_i;U_i|X,W,T)  & \ge
\left[g((\tilde{D} + \delta)^{1/\ell}) + 2 \delta \log
\frac{\delta}{5}\right]\cdot \Pr((W,T) \in S) \\
& \ge \left[g((\tilde{D} + \delta)^{1/\ell}) + 2 \delta \log
\frac{\delta}{5}\right]\left(1-\delta \right) \\
& = g((\tilde{D} + \xi(\tilde{D},\delta))^{1/\ell})
\end{align*}
for some continuous $\xi \ge 0$ satisfying $\xi(\tilde{D},0)=0$. It follows from this and constraint \emph{(iii)} of the
lemma that $g(D^{1/\ell}) \ge g((\tilde{D} + \xi(\tilde{D},\delta))^{1/\ell})$, and from the monotonicity of $g(D^{1/\ell})$ in $D$ (Corollary~\ref{convexcor}), we obtain
\begin{equation*}
\tilde{D} + \xi(\tilde{D},\delta) \ge D,
\end{equation*}
and thus
\begin{equation*}
\tilde{D} \ge D - \xi(\tilde{D},\delta).
\end{equation*}
This completes the proof.

\section*{Acknowledgment}
The authors would like to thank Chao Tian for suggestions and helpful discussions regarding the use of systematic MDS codes in the achievability scheme in Section~\ref{sec:worst-case-achievability}.

\end{document}